\documentclass[UKenglish,a4paper,pdfa,cleveref,thm-restate]{lipics-v2021}
\usepackage[usenames,dvipsnames]{xcolor}
\usepackage{booktabs,csquotes,mathtools,xfrac}

\graphicspath{{./figures/}}

\newcommand*{\R}{\mathbb{R}}
\newcommand*{\bO}{\mathcal{O}}
\newcommand*{\eps}{\varepsilon}
\newcommand*{\C}{\mathcal{C}}
\newcommand*{\B}{\mathcal{B}}
\newcommand*{\U}{\mathcal{U}}
\newcommand*{\Rays}{\mathcal{R}}
\newcommand*{\Arr}{\mathcal{A}}
\newcommand*{\Lns}{\mathcal{L}}

\newcommand*{\seg}[1]{\overline{#1}}
\newcommand*{\queryseg}{\seg{pq}}
\newcommand*{\SetSize}[1]{\lvert #1 \rvert}
\newcommand*{\DSPartRef}[1]{\textcolor{lipicsGray}{\normalfont\sffamily\bfseries\mathversion{bold}#1}}
\newcommand*{\spds}{GHDS}

\newcommand*{\conetop}{C_\textnormal{top}}
\newcommand*{\tadj}{T_\textnormal{adj}}

\newcommand*{\cred}{\C_\textnormal{red}}
\newcommand*{\cblue}{\C_\textnormal{blue}}
\newcommand*{\cpurple}{\C_\textnormal{purple}}
\newcommand*{\corange}{\C_\textnormal{orange}}
\newcommand*{\cgrey}{\C_\textnormal{grey}}
\newcommand*{\cgreen}{\C_\textnormal{green}}
\newcommand*{\ctop}{\C_\textnormal{top}}

\newcommand*{\norange}{N_\textnormal{orange}}
\newcommand*{\ngrey}{N_\textnormal{grey}}
\newcommand*{\ngreen}{N_\textnormal{green}}
\newcommand*{\borange}{\B_\textnormal{orange}}
\newcommand*{\bgrey}{\B_\textnormal{grey}}
\newcommand*{\bgreen}{\B_\textnormal{green}}

\DeclareMathOperator{\polylog}{polylog}

\title{Segment Visibility Counting Queries in Polygons}
\author{Kevin Buchin}{Department of Computer Science, TU Dortmund, Germany\and\url{https://www.win.tue.nl/~kbuchin/}}{kevin.buchin@tu-dortmund.de}{https://orcid.org/0000-0002-3022-7877}{}
\author{Bram Custers}{Department of Mathematics and Computer Science, TU Eindhoven, The Netherlands}{b.a.custers@tue.nl}{https://orcid.org/0000-0001-9342-319X}{Supported by the Dutch Research Council (NWO) under the project number 628.011.005.}
\author{Ivor van der Hoog}{Department of Applied Mathematics and Computer Science, TU Denmark, Copenhagen, Denmark}{vanderhoog@gmail.com}{}{Supported by the Dutch Research Council (NWO) under the project number 614.001.504.}
\author{Maarten Löffler}{Department of Information and Computing Sciences, Utrecht University, The Netherlands\and\url{https://webspace.science.uu.nl/~loffl001/}}{m.loffler@uu.nl}{}{Partially supported by the Dutch Research Council (NWO) under the project numbers 614.001.504 and 628.011.005.}
\author{Aleksandr Popov}{Department of Mathematics and Computer Science, TU Eindhoven, The Netherlands\and\url{https://www.win.tue.nl/~apopov/}}{a.popov@tue.nl}{https://orcid.org/0000-0002-0158-1746}{Supported by the Dutch Research Council (NWO) under the project number 612.001.801.}
\author{Marcel Roeloffzen}{Department of Mathematics and Computer Science, TU Eindhoven, The Netherlands\and\url{https://www.win.tue.nl/~mroeloff/}}{m.j.m.roeloffzen@tue.nl}{https://orcid.org/0000-0002-1129-461X}{Supported by the Dutch Research Council (NWO) under the project number 628.011.005.}
\author{Frank Staals}{Department of Information and Computing Sciences, Utrecht University, The Netherlands\and\url{https://fstaals.net/}}{f.staals@uu.nl}{}{}
\authorrunning{K. Buchin, B. Custers, I. v/d Hoog, M. Löffler, A. Popov, M. Roeloffzen, F. Staals}
\hideLIPIcs
\nolinenumbers

\begin{document}
\maketitle

\begin{abstract}
Let \(P\) be a simple polygon with \(n\) vertices, and let \(A\) be a set of
\(m\) points or line segments inside \(P\).
We develop data structures that can efficiently count the number of objects from
\(A\) that are visible to a query point or a query segment.
Our main aim is to obtain fast, \(\bO(\polylog nm)\), query times, while using
as little space as possible.
In case the query is a single point, a simple visibility-polygon-based solution
achieves \(\bO(\log nm)\) query time using \(\bO(nm^2)\) space.
In case \(A\) also contains only points, we present a smaller,
\(\bO(n + m^{2 + \eps}\log n)\)-space, data structure based on a hierarchical
decomposition of the polygon.
Building on these results, we tackle the case where the query is a line segment
and \(A\) contains only points.
The main complication here is that the segment may intersect multiple regions of
the polygon decomposition, and that a point may see multiple such pieces.
Despite these issues, we show how to achieve \(\bO(\log n\log nm)\) query time
using only \(\bO(nm^{2 + \eps} + n^2)\) space.
Finally, we show that we can even handle the case where the objects in \(A\) are
segments with the same bounds.
\end{abstract}
  
\section{Introduction}\label{sec:intro}
Let \(P\) be a simple polygon with \(n\) vertices, and let \(A\) be a set of
\(m\) points or line segments inside \(P\).
We develop efficient data structures for \emph{visibility counting queries} in
which we wish to report the number of objects from \(A\) visible to some
(constant-complexity) query object \(Q\).
An object \(X\) in \(A\) is \emph{visible} from \(Q\) if there is a line segment
connecting \(X\) and \(Q\) contained in \(P\).
We are mostly interested in the case when \(Q\) is a point or a line segment.
Our aim is to obtain fast, \(\bO(\polylog nm)\), query times, using as little
space as possible.
Our work is motivated by problems in movement analysis where we have sets of
entities, for example, an animal species and their predators, moving in an
environment, and we wish to determine if there is mutual visibility between the
entities of different sets.
We also want to quantify `how much' the sets can see each other.
Assuming we have measurements at certain points in time, solving the mutual
visibility problem between two such times reduces to counting visibility between
line segments (for moving entities) and points (for static objects or entities).

\subparagraph*{Related work.}
Computing visibility is a classical problem in computational
geometry~\cite{ghosh07,orourke87}.
Algorithms for efficiently testing visibility between a pair of points, for
computing visibility polygons~\cite{elgindy81,joe87,lee83}, and for constructing
visibility graphs~\cite{overmars88} have been a topic of study for over thirty
years.
There is even a host of work on computing visibility on terrains and in other
three-dimensional environments~\cite{agarwal93rayshoot,berg94}.
For many of these problems, the data structure version of the problem has also
been considered.
In these versions, the polygon is given up front, and the task is to store it so
that we can efficiently query whether or not a pair of points \(p, q\) is
mutually visible~\cite{chazelle89,guibas87,hershberger95}, or report the entire
visibility polygon \(V(q)\) of \(q\)~\cite{aronov02}.
In particular, when \(P\) is a simple polygon with \(n\) vertices, the former
type of queries can be answered optimally---in \(\bO(\log n)\) time using
linear space~\cite{hershberger95}.
Answering the latter type of queries can be done in
\(\bO(\log^2 n + \SetSize{V(q)})\) time using \(\bO(n^2)\)
space~\cite{aronov02}.
The visibility polygon itself has complexity \(\bO(n)\)~\cite{elgindy81}.

Computing the visibility polygon of a line segment has been considered, as well.
When the polygon modelling the environment is simple, the visibility polygon,
called a \emph{weak visibility polygon,} denoted \(V(\queryseg)\) for a line
segment \(\queryseg\), still has linear complexity, and can be computed in
\(\bO(n)\) time~\cite{guibas87}.
Chen and Wang~\cite{chen15weak} consider the data structure version of the
problem: they describe a linear-space data structure that can be queried in
\(\bO(\SetSize{V(\queryseg)}\log n)\) time, and an \(\bO(n^3)\)-space data
structure that can be queried in \(\bO(\log n + \SetSize{V(\queryseg)})\) time.

Computing the visibility polygon of a line segment \(\queryseg\) allows us to
answer whether an entity moving along \(\queryseg\) can see a particular fixed
point \(r\), i.e.\@ there is a time at which the moving entity can see \(r\) if
and only if \(r\) lies inside \(V(\queryseg)\).
If the point \(r\) may also move, it is not necessarily true that the entity can
see \(r\) if the trajectory of \(r\) intersects \(V(\queryseg)\).
Eades et al.~\cite{eades20} present data structures that can answer such queries
efficiently.
In particular, they present data structures of size \(\bO(n\log^5 n)\) that can
answer such a query in time \(\bO(n^{\sfrac{3}{4}}\log^3 n)\).
They present results even in case the polygon has holes.
Aronov et al.~\cite{aronov02} show that we can also efficiently maintain the
visibility polygon of an entity as it is moving.

Visibility counting queries have been studied before, as well.
Bose et al.~\cite{bose02} studied the case where, for a simple polygon and a
query point, the number of visible polygon edges is reported.
The same problem has been considered for weak visibility from a query
segment~\cite{bygi15}.
For the case of a set of disjoint line segments and a query point, approximation
algorithms exist~\cite{alipour15,gudmundsson10,suri86}.
In contrast to these settings, we wish to count visible line segments with
visibility obstructed by a simple polygon (other than the line segments).
Closer to our setting is the problem of reporting all pairs of visible points
within a simple polygon~\cite{ben-moshe04}.

\subparagraph*{Results and organisation.}
Our goal is to efficiently count the number of objects, in particular, line
segments or points, in a set \(A\) that are visible to a query object \(Q\).
We denote this number by \(C(Q, A)\).
Given \(P\), \(A\), and \(Q\), we can easily compute \(C(Q, A)\) in optimal
\(\bO(n + m\log n)\) time (see \cref{lem:algorithm_oneshot}).
We are mostly interested in the data structure version of the problem, in which
we are given the polygon \(P\) and the set \(A\) in advance, and we wish to
compute \(C(Q, A)\) efficiently once we are given the query object \(Q\).
We show that we can indeed answer such queries efficiently, that is, in
polylogarithmic time in \(n\) and \(m\).
The exact query times and the space usage and preprocessing times depend on the
type of the query object and the type of objects in \(A\).
See \cref{tab:results} for an overview.

\newcolumntype{Y}[1]{>{\hsize=#1\hsize\linewidth=\hsize}X}
\begin{table}[tb]
\centering
\caption{Results in this paper.
\(\bullet\) and \(\slash\) denote points and line segments, respectively.}
\begin{tabularx}{\linewidth}{c c Y{.75} Y{1.5} Y{.75} c}
\toprule
\multirow{2}{*}{\(A\)} & \multirow{2}{*}{\(Q\)} &
    \multicolumn{3}{c}{Data structure} & \multirow{2}{*}{Section}\\
& & Space & Preprocessing & Query & \\
\midrule
\(\bullet\) & \(\bullet\) & \(\bO(nm^2)\) & \(\bO(nm\log n + nm^2)\) &
    \(\bO(\log nm)\) & \ref{sec:arrangement}\\
& & \(\bO(n + m^{2 + \eps} \log n)\) &
    \(\bO(n + m \log^2 n + m^{2 + \eps} \log n)\) & \(\bO(\log n \log nm)\) &
    \ref{sec:point_point}\\
\(\slash\) & \(\bullet\) & \(\bO(nm^2)\) & \(\bO(nm\log n + nm^2)\) &
    \(\bO(\log nm)\) & \ref{sec:arrangement}\\
\(\bullet\) & \(\slash\) & \(\bO(n^2 + nm^{2 + \eps})\) &
    \(\bO(n^2\log m + nm\log n + nm^{2 + \eps})\) & \(\bO(\log n\log nm)\) &
    \ref{sec:point_segment}\\
\(\slash\) & \(\slash\) & \(\bO(n^2 + nm^{2 + \eps})\) &
    \(\bO(n^2\log m + nm\log n + nm^{2 + \eps})\) & \(\bO(\log n\log nm)\) &
    \ref{sec:segment_segment}\\
\bottomrule
\end{tabularx}
\label{tab:results}
\end{table}

In \cref{sec:point}, we consider the case where the query object is a point.
We show how to answer queries efficiently using the arrangement of all (weak)
visibility polygons, one for each object in \(A\).
As Bose et al.~\cite[Section~6.2]{bose02} argued, such an arrangement has
complexity \(\Theta(nm^2)\) in the worst case.

\begin{restatable}{theorem}{theoremArrangement}\label{thm:arrangement}
Let \(P\) be a simple polygon with \(n\) vertices, and let \(A\) be a set of
\(m\) points or line segments inside \(P\).
In \(\bO(nm^2 + nm\log n)\) time, we can build a data structure of size
\(\bO(nm^2)\) that can report the number of points or segments in \(A\) visible
from a query point \(q\) in \(\bO(\log nm)\) time.
\end{restatable}

We then show that if the objects in \(A\) are points, we can significantly
decrease the required space.
We argue that we do not need to construct the visibility polygons of all points
in \(A\), thus avoiding an \(\bO(nm)\) term in the space and preprocessing time.
We use a hierarchical decomposition of the polygon and the fact that the
visibility of a point \(a \in A\) in a subpolygon into another subpolygon is
described by a single constant-complexity cone.
We then obtain the following result.
Here and in the rest of the paper, \(\eps > 0\) is an arbitrarily small
constant.

\begin{restatable}{theorem}{theoremPointPoint}\label{thm:point_point}
Let \(P\) be a simple polygon with \(n\) vertices, and let \(A\) be a set of
\(m\) points inside \(P\).
In \(\bO(n + m^{2 + \eps} \log n + m \log^2 n)\) time, we can build a data
structure of size \(\bO(n + m^{2 + \eps} \log n)\) that can report the number of
points from \(A\) visible from a query point \(q\) in \(\bO(\log n \log nm)\)
time.
\end{restatable}

In \cref{sec:point_segment}, we turn our attention to the case where the query
object is a line segment \(\queryseg\) and the objects in \(A\) are points.
One possible solution in this scenario would be to store the visibility polygons
for the points in \(A\) so that we can count the number of such polygons stabbed
by the query segment.
However, since these visibility polygons have total complexity \(\bO(nm)\) and
the query may have an arbitrary orientation, a solution achieving
polylogarithmic query time will likely use at least \(\Omega(n^2m^2)\)
space~\cite{agarwal93spacepart,agarwal96,gupta95}.
So, we again use an approach that hierarchically decomposes the polygon to limit
the space usage.
Unfortunately, testing visibility between the points in \(A\) and the query
segment is more complicated in this case.
Moreover, the segment can intersect multiple regions of the decomposition, so we
have to avoid double counting.
All of this makes the problem significantly harder.
We manage to overcome these difficulties using careful geometric arguments and
an inclusion--exclusion-style counting scheme.
This leads to the following result, saving at least a linear factor compared to
an approach based on stabbing visibility polygons:

\begin{restatable}{theorem}{theoremPointSegment}\label{thm:point_segment}
Let \(P\) be a simple polygon with \(n\) vertices, and let \(A\) be a set of
\(m\) points inside \(P\).
In time \(\bO(nm^{2 + \eps} + nm \log n + n^2\log m)\), we can build a data
structure of size \(\bO(nm^{2 + \eps} + n^2)\) that can report the number of
points from \(A\) visible from a query segment \(\queryseg\) in
\(\bO(\log n \log nm)\) time.
\end{restatable}

In \cref{sec:segment_segment}, we show that we can extend these arguments even
further and solve the scenario where the objects in \(A\) are also line
segments.
Somewhat surprisingly, this does not impact the space or time complexity of the
data structure.
(Note that in this setting, visibility between segments does not represent
visibility between moving points; refer to \cref{sec:extra} for a discussion of
that problem.)

\begin{restatable}{theorem}{theoremSegmentSegment}
\label{thm:segment_segment}
Let \(P\) be a simple polygon with \(n\) vertices, and let \(A\) be a set of
\(m\) segments inside \(P\).
In time \(\bO(nm^{2 + \eps} + nm \log n + n^2\log m)\), we can build a data
structure of size \(\bO(nm^{2 + \eps} + n^2)\) that can report the number of
points from \(A\) visible from a query segment \(\queryseg\) in
\(\bO(\log n \log nm)\) time.
\end{restatable}

Finally, in \cref{sec:extra}, we discuss some extensions of our results.
We generalise the approach of \cref{sec:segment_segment} to the case where the
objects in \(A\) are simple polygons.
We consider some query variations and show that we can compute the pairwise
visibility of two sets of objects---that is, solve the problem that motivated
this work---in time subquadratic in the number of objects.

\section{Preliminaries}\label{sec:prelims}
In this \lcnamecref{sec:prelims}, we review some basic terminology and tools we
use to build our data structures.

\subparagraph*{Visibility in a simple polygon.}
We refer to the parts of the polygon \(P\) that can be seen from some point
\(p \in P\) as its \emph{visibility polygon,} denoted \(V(p)\).
The visibility polygon has complexity \(\bO(n)\)~\cite{elgindy81}.
We can also construct a visibility polygon for a line segment \(\queryseg\),
denoted \(V(\queryseg)\), which is the union of the visibility polygons of all
points on \(\queryseg\); it is referred to as a \emph{weak visibility polygon.}
Such a polygon still has complexity \(\bO(n)\)~\cite{guibas87}.

\begin{lemma}\label{lem:algorithm_oneshot}
Let \(P\) be a simple polygon with \(n\) vertices, and let \(A\) be a set of
\(m\) points or line segments inside \(P\).
We can compute the number \(C(Q, A)\) of objects from \(A\) visible to a point
or line segment \(Q\) in time \(\bO(n + m\log n)\).
\end{lemma}
\begin{proof}
If \(A\) is a set of points, it suffices to compute the visibility polygon of
\(Q\) and preprocess it for \(\bO(\log n)\)-time point location queries.
Both preprocessing steps take linear time~\cite{guibas87,kirkpatrick83}, and
querying takes \(\bO(m\log n)\) time in total.
In case \(A\) consists of line segments, we can similarly test if a segment of
\(A\) is visible when one of the endpoints is visible.
We also need to count the number of visible segments whose endpoints lie outside
of \(V(Q)\).
We can also do this in \(\bO(n + m \log n)\) time by computing a sufficiently
large bounding box \(B\) and constructing a \(\bO(\log n)\)-time ray shooting
data structure on \(B \setminus V(Q)\).
This allows us to test if a segment intersects \(V(Q)\) in \(\bO(\log n)\) time.
Since \(B \setminus V(Q)\) has only a single hole, we can turn it into a simple
polygon, build a ray shooting structure for simple
polygons~\cite{hershberger95}, and answer a query by \(\bO(1)\) ray shooting
queries.
\end{proof}

\begin{lemma}\label{lem:visibility-seg-intersect}
Given a visibility polygon \(V(p) \subseteq P\) for some point \(p \in P\) and a
line segment \(\seg{rs} \subset P\), if \(\seg{rs}\) and \(V(p)\) intersect,
their intersection is a single line segment.
\end{lemma}
\begin{proof}
Assume for contradiction that the intersection between \(V(p)\) and \(\seg{rs}\)
consists of multiple, possibly degenerate, line segments, \(S_1, \dots, S_k\)
for some \(k > 1\).
Take some points \(q_i\) and \(q_{i + 1}\) on consecutive segments \(S_i\) and
\(S_{i + 1}\).
Consider the line segments \(\seg{pq_i}\) and \(\seg{pq_{i + 1}}\).
By definition of the visibility polygon, these segments are inside \(P\).
Since \(\seg{rs}\) is inside \(P\), the segment \(\seg{q_i q_{i + 1}}\) is also
inside \(P\).
Since \(P\) is simple, it must then hold that the interior of the triangle \(T\)
with vertices \(p\), \(q_i\), and \(q_{i + 1}\) is also inside \(P\).
More precisely, \(T\) cannot contain any of the boundary of \(P\).
Now consider a line segment \(\seg{pq_o}\) for a point \(q_o\) between segments
\(S_i\) and \(S_{i + 1}\) on \(\seg{rs}\).
Since its endpoint \(q_o\) is outside \(V(p)\), the line segment must cross the
boundary of \(P\) inside \(T\).
This contradicts the previous claim that \(T\) is empty; thus, it must be that
the intersection between \(V(p)\) and \(\seg{rs}\) is a line segment if they
intersect.
\end{proof}

A \emph{cone} is a subspace of the plane that is enclosed by two rays starting
at some point \(p\), called the \emph{apex} of the cone.
Let \(D\) be a diagonal of \(P\) and let \(P_L\) and \(P_R\) be the two
subpolygons we obtain when splitting \(P\) with \(D\).
Consider some point \(p\) in \(P_L\).
We define the \emph{visibility cone} of \(p\) \emph{through} \(D\), denoted
\(V(p, D, P_L)\), as the collection of rays starting at \(p\) that intersect
\(D\) and do not intersect the boundary of \(P_L\), except at \(D\) (see
\cref{fig:cone_illustration}).

\begin{figure}[tb]
\centering
\includegraphics{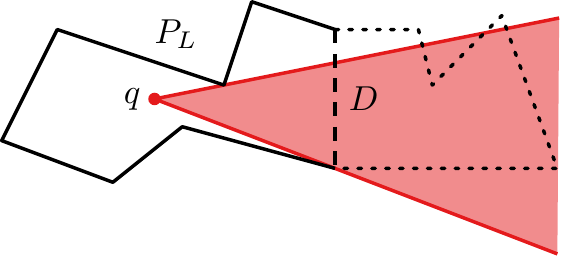}
\caption{The filled shape is the cone \(V(q, D, P_L)\).}
\label{fig:cone_illustration}
\end{figure}

\begin{corollary}\label{cor:cone_continuous}
The intersection between the visibility cone \(V(p, D, P_L)\) and the diagonal
\(D\) is a single line segment or empty.
\end{corollary}
\begin{proof}
The visibility cone is clearly a subset of the visibility polygon \(V(p)\).
Since the intersection between \(D\) and the visibility polygon is a single line
segment (or empty) by \cref{lem:visibility-seg-intersect}, the same must hold
for the visibility cone.
\end{proof}

\subparagraph*{Cutting trees.}
Suppose we want to preprocess a set \(\Lns\) of \(m\) lines in the plane so that
given a query point \(q\), we can count the number of lines below the query
point.
Let \(r \in [1, m]\) be a parameter; then a \emph{\((\sfrac{1}{r})\)-cutting} of
\(\Lns\) is a subdivision of the plane with the property that each cell is
intersected by at most \(\sfrac{m}{r}\) lines~\cite{chazelle93}.
If \(q\) lies in a certain cell of the cutting, we know, for all lines that do
not cross the cell, whether they are above or below \(q\), and so we can store
the count with the cell, or report the lines in a precomputed \emph{canonical
subset;} for the lines that cross the cell, we can recurse.
The data structure that performs such a query is called a \emph{cutting tree;}
it can be constructed in \(\bO(m^{2 + \eps})\) time, uses \(\bO(m^{2 + \eps})\)
space, and supports answering the queries in time \(\bO(\log m)\), for any
constant \(\eps > 0\).
Intuitively, the parameter \(r\) here determines the trade-off between the
height of the recursion tree and the number of nodes for which a certain line in
\(\Lns\) is relevant.
If we pick \(r = m\), the \((\sfrac{1}{r})\)-cutting of \(\Lns\) is just the
arrangement of \(\Lns\).
The bounds above are based on picking \(r \in \bO(1)\), so the height of the
recursion tree is \(\bO(\log m)\).
This approach follows the work of Clarkson~\cite{clarkson87}, with
Chazelle~\cite{chazelle93} obtaining the bounds above by improving the cutting
construction.

An obvious benefit of this approach over just constructing the arrangement on
\(\Lns\) and doing point location in that arrangement is that using cuttings,
we can obtain \(\bO(\log m)\) canonical subsets and perform nested queries on
them without an explosion in storage required; the resulting data structure is
called a \emph{multilevel cutting tree.}
Specifically, we can query with \(k\) points and a direction associated with
each point (above or below) and return the lines of \(\Lns\) that pass on the
correct side (above or below) of all \(k\) query points.
If we pick \(r \in \bO(1)\) and nest \(k\) levels in a \(k\)-level cutting tree,
we get the same construction time and storage bounds as for a regular cutting
tree; but the query time is now \(\bO(\log^k m)\).
Chazelle et al.~\cite{chazelle92fast} show that if we set
\(r = n^{\sfrac{\eps}{2}}\), each level of a multilevel cutting tree is a
constant-height tree, so the answer to the query can be represented using only
\(\bO(1)\) canonical subsets.
The space used and the preprocessing time remains \(\bO(m^{2 + \eps})\).

\begin{lemma}\label{lem:multilevel-cutting-tree-2}
Let \(\Lns\) be a set of \(m\) lines and let \(k\) be a constant.
Suppose we want to answer the following query: given \(k\) points and associated
directions (above or below), or \(k\) vertical rays, find the lines in \(\Lns\)
that lie on the correct side of all \(k\) points (or intersect all \(k\) rays).
In time \(\bO(m^{2 + \eps})\), we can construct a data structure using
\(\bO(m^{2 + \eps})\) storage that supports such queries.
The lines are returned as \(\bO(1)\) canonical subsets, and the query time is
\(\bO(\log m)\).
\end{lemma}

Dualising the problem in the usual way, we can alternatively report or count
points from the set \(A\) that lie in a query half-plane; or in the intersection
of several half-planes, using a multilevel cutting tree.

\begin{lemma}\label{lem:multilevel-cutting-tree-1}
Let \(A\) be a set of \(m\) points and let \(k\) be a constant.
In time \(\bO(m^{2 + \eps})\), we can construct a data structure using
\(\bO(m^{2 + \eps})\) storage that returns \(\bO(1)\) canonical subsets with the
points in \(A\) that lie in the intersection of the \(k\) query half-planes in
time \(\bO(\log m)\).
\end{lemma}

We can use these basic results to resolve more complicated queries; the
techniques are similar and are shown in \cref{lem:mutual_cones}.
See \cref{fig:cutting_tree} for an illustration.

\begin{lemma}\label{lem:mutual_cones}
Let \(A\) be a set of \(m\) points, each of them an apex of a cone; and let the
query point be \(q\), again an apex of a cone \(C_q\).
In time \(\bO(m^{2 + \eps})\), we can construct a data structure using
\(\bO(m^{2 + \eps})\) storage that returns a representation of the points in
\(A' \subseteq A\), such that for any \(p \in A'\), \(q\) lies in the cone of
\(p\) and \(p \in C_q\).
The points are returned as \(\bO(1)\) canonical subsets, and the query time is
\(\bO(\log m)\).
Alternatively, the points can be counted.
\end{lemma}
\begin{proof}
We can construct a four-level cutting tree; the first two levels can select the
nodes that represent points from \(A\) lying in \(C_q\).
Note that to select the points that lie in \(C_q\), we need to perform two
consecutive half-plane queries, as \(C_q\) is an intersection of two half-planes
that meet at point \(q\).
We can use \cref{lem:multilevel-cutting-tree-1} to handle these; note that every
time we get a constant number of canonical subsets, so any new point location
queries can be done in \(\bO(\log m)\) time on each level.
After two levels, we get \(\bO(1)\) canonical subsets.
The next two levels handle the other condition: select the points whose cones
contain \(q\).
This can be done by checking that \(q\) lies below the upper boundaries of the
cones and that \(q\) lies above the lower boundaries of the cones.
Again, we need to do point location queries on each level and for each canonical
subset; we can use \cref{lem:multilevel-cutting-tree-2} to see that we still
have a constant number of those.
Overall, we do a constant number of point location queries and go down a
four-level data structure, where every level is a constant-depth tree.
Therefore, the query overall takes \(\bO(\log m)\) time.
As stated previously, adding the levels does not increase the storage or the
preprocessing time requirements.
\end{proof}

\begin{figure}[tb]
\centering
\includegraphics{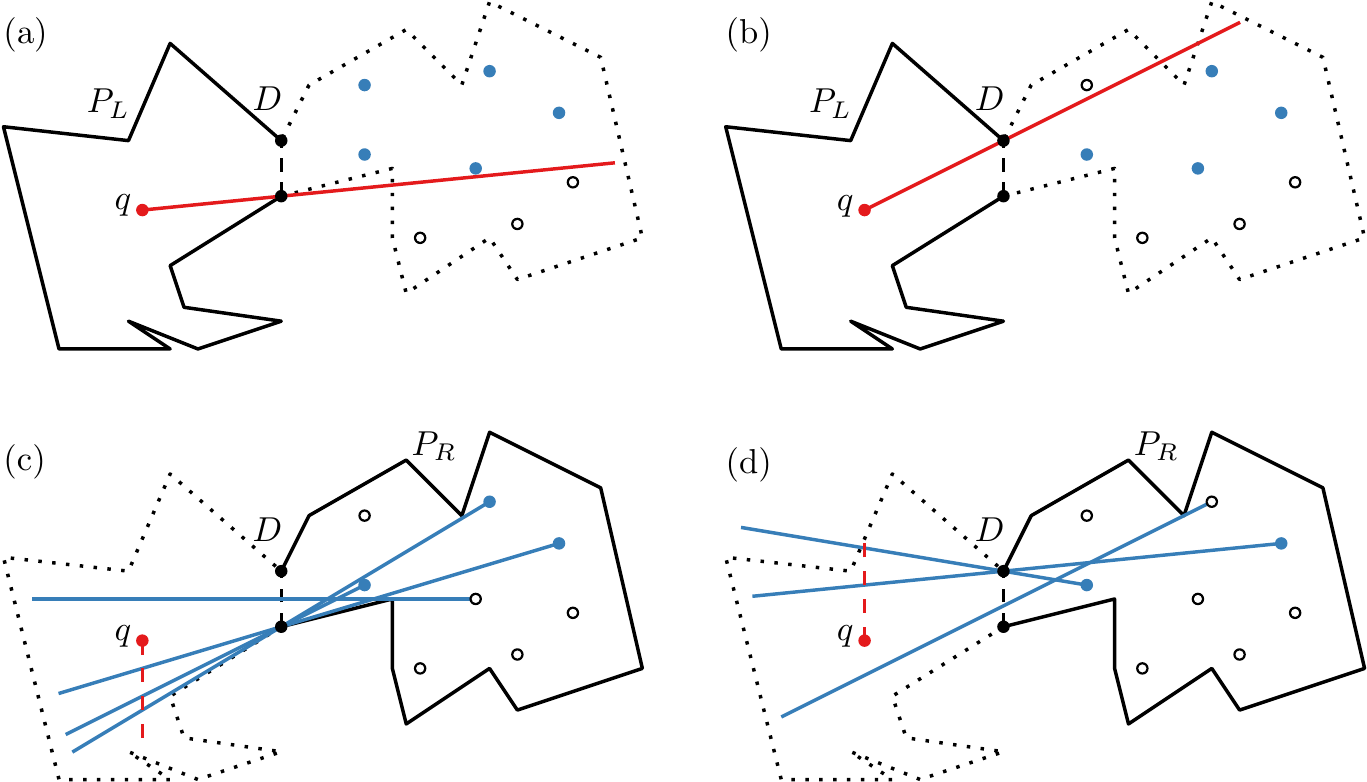}
\caption{A query in a multilevel cutting tree, top left to bottom right.
The query point is red; the selected points of \(A\) are blue.
Black outline shows the relevant part of the polygon.
(a,~b)~We select points in \(A\) above (resp.\@ below) the right (resp.\@ left)
cone boundary of \(q\).
(c,~d)~We refine by taking points whose left (resp.\@ right) cone boundary is
below (resp.\@ above) \(q\).}
\label{fig:cutting_tree}
\end{figure}

\begin{lemma}\label{lem:separated-cones-seg-ds}
Let \(L\) be a vertical line and let \(A\) be a set of \(m\) cones starting left
of \(L\) and whose top and bottom rays intersect \(L\).
In time \(\bO(m^{2 + \eps})\), we can construct two two-level cutting trees for
\(A\) of total size \(\bO(m^{2 + \eps})\), so that for a query segment
\(\queryseg\) that is fully to the right of \(L\), we can count the number of
cones that contain or intersect \(\queryseg\) in \(\bO(\log m)\) time.
\end{lemma}
\begin{proof}
A cone \(C \in A\) partitions the space to the right of \(L\) in three regions:
the regions above and below \(C\) and the region inside \(C\).
Segment \(\queryseg\) does not intersect \(C\) if it is contained in either the
top or the bottom region.
This is exactly when either both points of \(\queryseg\) are above the
supporting line of the upper boundary of \(C\) or when both are below the
supporting line of the lower boundary of \(C\).
Hence, if we store the supporting lines of the upper and lower boundaries of
\(A\) in two two-level cutting trees, similarly to \cref{lem:mutual_cones}, we
can count the number of cones that are not visible for \(\queryseg\).
By storing the total number of cones, we can now determine the number of visible
cones.
\end{proof}

\begin{lemma}\label{lem:count-segment-line-intersections}
Let \(\Lns\) be a set of \(m\) lines and \(\queryseg\) a query line segment.
We can store \(\Lns\) in a multilevel cutting tree, using \(\bO(m^{2 + \eps})\)
space and preprocessing time, so that we can count the number of lines in
\(\Lns\) intersected by \(\queryseg\) in time \(\bO(\log m)\).
\end{lemma}

\subparagraph*{Polygon decomposition.}
For a simple polygon \(P\) on \(n\) vertices, Chazelle~\cite{chazelle82}
shows that we can construct a balanced hierarchical decomposition of \(P\) by
recursively splitting the polygon into two subpolygons of approximately
equal size.
The polygon is split on diagonals between two vertices of the polygon.
The recursion stops when reaching triangles.
The decomposition can be computed in \(\bO(n)\) time and stored using \(\bO(n)\)
space in a balanced binary tree.

\subparagraph*{Hourglasses and the shortest path data structure.}
An \emph{hourglass} for two segments \(\seg{pq}\) and \(\seg{rs}\) in a simple
polygon \(P\) is the union of geodesic shortest paths in \(P\) from a point on
\(\seg{pq}\) to a point on \(\seg{rs}\)~\cite{guibas89}.
If the upper chain and lower chain of an hourglass share vertices, it is
\emph{closed,} otherwise it is \emph{open.}
A \emph{visibility glass} is a subset of the hourglass consisting of all line
segments between a point on \(\seg{pq}\) and a point on
\(\seg{rs}\)~\cite{eades20}.

Guibas and Hershberger~\cite{guibas89}, with later
improvements~\cite{hershberger91}, describe a data structure to compute shortest
paths in a simple polygon \(P\).
They use the polygon decomposition by Chazelle~\cite{chazelle82} and also store
hourglasses between the splitting diagonals of the decomposition.
The data structure uses \(\bO(n)\) storage and preprocessing time and can answer
the following queries in \(\bO(\log n)\) time:
\begin{description}
\item[Shortest path query.]
Given points \(p, q \in P\), return the geodesic shortest path between \(p\) and
\(q\) in \(P\) as a set of \(\bO(\log n)\) nodes of the decomposition.
The shortest path between \(p\) and \(q\) is a concatenation of the polygonal
chains of the boundaries of the (open or closed) hourglasses in these
\(\bO(\log n)\) nodes together with at most \(\bO(\log n)\) segments called
\emph{bridges} connecting two hourglass boundaries.
\item[Segment location query.]
Given a segment \(\seg{pq}\), return the two leaf triangles containing \(p\) and
\(q\) in the decomposition and the \(\bO(\log n)\) pairwise disjoint open
hourglasses such that the two leaf triangles and hourglasses fully cover
\(\seg{pq}\).
We refer to the returned structure as the \emph{polygon cover} of \(\seg{pq}\).
\item[Cone query.]
Given a point \(s\) and a line segment \(\seg{pq}\) in \(P\), return the
visibility cone from \(s\) through \(\seg{pq}\).
This can be done by getting the shortest paths from \(s\) to \(p\) and \(q\)
and taking the two segments closest to \(s\) to extend them into a cone.
\end{description}

\section{Point Queries}\label{sec:point}
In this \lcnamecref{sec:point}, given a set \(A\) of \(m\) points in a simple
polygon \(P\) on \(n\) vertices, we count the points of \(A\) that are in the
visibility polygon of a query point \(q \in P\).
We present two solutions: (i)~a simple arrangement-based approach that is also
applicable to the case where \(A\) contains line segments, which achieves a very
fast query time at the cost of large storage and preprocessing time; and (ii)~a
cutting-tree-based approach with query times slower by a factor of
\(\bO(\log n)\), but with much better storage requirements and preprocessing
time.

\subsection{Point Location in an Arrangement}\label{sec:arrangement}
We first consider a straightforward approach that relies on the following
observation: the number of objects in \(A\) visible from a query point \(q\) is
equal to the number of (weak) visibility polygons of the objects in \(A\)
stabbed by \(q\).
Hence, we can construct all (weak) visibility polygons of the objects in \(A\)
and compute the arrangement \(\Arr\) of the edges of these polygons.
For each cell \(C\) in the arrangement, we store the number of visibility
polygons that contain \(C\).
Then a point location query for \(q\) yields the number of visible objects in
\(A\).

Computing the visibility polygons takes \(\bO(nm)\) time, and constructing the
arrangement using an output-sensitive line segment intersection algorithm takes
\(\bO(nm\log nm + \SetSize{\Arr})\) time~\cite{chazelle92intls}, where
\(\SetSize{\Arr}\) is the number of vertices of \(\Arr\).
Building a point location structure on \(\Arr\) for
\(\bO(\log \SetSize{\Arr})\)-time point location queries takes
\(\bO(\SetSize{\Arr})\) time~\cite{kirkpatrick83}.
The space usage is \(\bO(\SetSize{\Arr})\).
As Bose et al.~\cite{bose02} show, the complexity of \(\Arr\) is
\(\Theta(nm^2)\) in the worst case.

\theoremArrangement*

\subsection{Hierarchical Decomposition}\label{sec:point_point}
To design a data structure that uses less storage than that of
\cref{sec:arrangement}, we observe that if we subdivide the polygon, we can
count the number of visible objects by summing up the number of visible objects
residing in the cells of the subdivision.
To efficiently compute these counts, we use the polygon decomposition approach,
as described in \cref{sec:prelims}.
With each split in our decomposition, we store data structures that can
efficiently count the number of visible objects in the associated subpolygon.

\subparagraph*{Cone containment.}
Let us solve the following problem first.
We are given a simple polygon \(P\) and a (w.l.o.g.) vertical diagonal \(D\)
that splits it into two simple polygons \(P_L\) and \(P_R\).
Furthermore, we are given a set \(A\) of \(m\) points in \(P_L\).
Given a query point \(q\) in \(P_R\), we want to count the points in \(A\) that
see \(q\).
We base our approach on the following observation.

\begin{figure}
\centering
\includegraphics{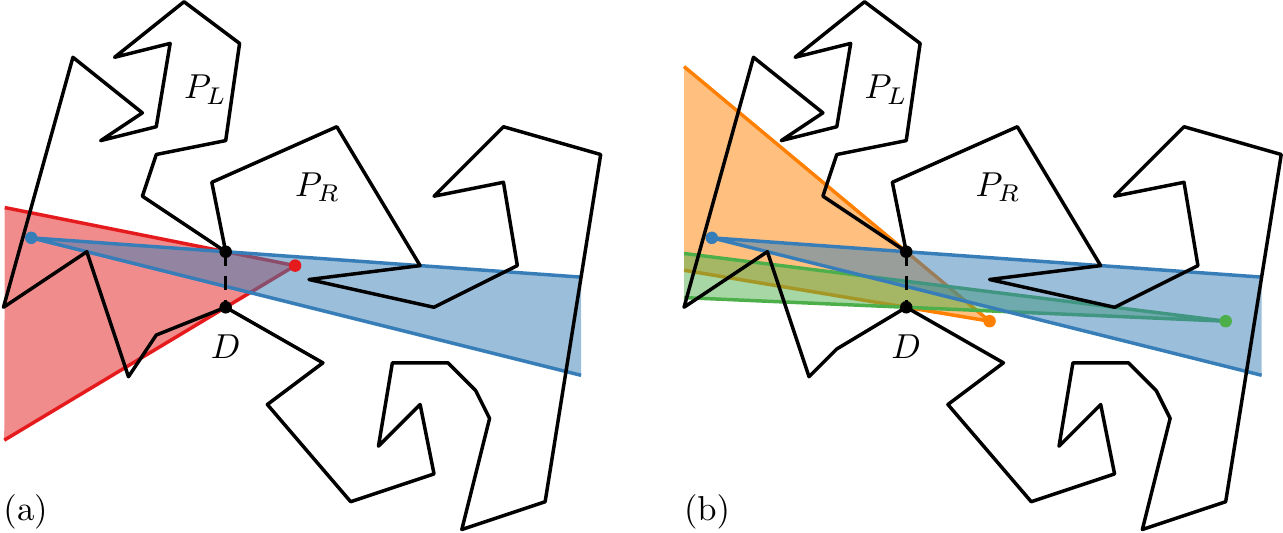}
\caption{Visibility cones (coloured regions) of (coloured) points w.r.t.\@ some
diagonal \(D\).
(a)~Blue and red are mutually visible.
(b)~Green and blue cannot see each other, nor can orange and blue.}
\label{fig:visibility-cone}
\end{figure}

\begin{lemma}\label{lem:cones}
Given a simple polygon \(P\) split into two simple polygons \(P_L\) and \(P_R\)
by a diagonal \(D\) between two vertices and two points \(p \in P_L\) and
\(q \in P_R\), consider the visibility cones \(V(p, D, P_L)\) and
\(V(q, D, P_R)\), i.e.\@ the cones from \(p\) and \(q\) through \(D\) into the
other subpolygons.
Point \(p\) sees \(q\) in \(P\) if and only if \(q \in V(p, D, P_L)\) and
\(p \in V(q, D, P_R)\).
\end{lemma}
\begin{proof}
First suppose that \(p \in V(q, D, P_R)\) and \(q \in V(p, D, P_L)\).
We need to show that \(p\) and \(q\) see each other, that is, that the line
segment \(\seg{pq}\) lies in \(P\).
Observe that both \(p\) and \(q\) lie in \(V(p, D, P_L)\), and \(V(p, D, P_L)\)
is convex, so \(\seg{pq}\) lies in \(V(p, D, P_L)\); symmetrically, \(\seg{pq}\)
lies in \(V(q, D, P_R)\).
Furthermore, note that since both cones are cones through \(D\), the segment
\(\seg{pq}\) must cross \(D\) at some point \(r\).
Then by construction of \(V(p, D, P_L)\), the segment \(\seg{pr}\) lies entirely
in \(P_L\); similarly, \(\seg{rq}\) lies entirely in \(P_R\).
As also the diagonal \(D\) lies in \(P\), we conclude that \(\seg{pq}\) lies in
\(P\).

Now suppose that \(p\) and \(q\) see each other in \(P\).
As they are on the opposite sides of the diagonal \(D\) and the polygon \(P\) is
simple, the line segment \(\seg{pq}\) must cross \(D\) at some point \(r\).
As \(\seg{pq}\) lies inside \(P\), clearly, \(\seg{pr}\) lies inside \(P_L\) and
\(\seg{rq}\) lies inside \(P_R\).
Then the visibility cone \(V(p, D, P_L)\) must include the ray from \(p\)
through \(r\), and so \(q\) is in \(V(p, D, P_L)\); symmetrically, \(p\) is in
\(V(q, D, P_R)\).
\end{proof}

\Cref{lem:cones} shows that to count the number of points of \(A\) that see
\(q\), it suffices to construct the cones from all points in \(A\) through \(D\)
and the cone from \(q\) through \(D\) and count the number of points from \(A\)
satisfying the condition of \cref{lem:cones} (see also
\cref{fig:visibility-cone}).
The cones from all \(p \in A\) into \(P_R\) can be precomputed, so only the
cone from \(q\) into \(P_L\) needs to be computed at query time.
The query of this type can be realised using a multilevel cutting tree, as
explained in \cref{lem:mutual_cones}.
We also still need to compute the cone \(V(q, D, P_R)\) at query time and
precompute the cones \(V(p, D, P_L)\) for all \(p \in A\); we shall handle this
later.

\subparagraph*{Decomposition.}
Let us return to the original problem.
To solve it, we can use the balanced polygon decomposition~\cite{chazelle82},
as discussed in \cref{sec:prelims}.
Following Guibas and Hershberger~\cite{guibas89,hershberger91}, we represent it
as a binary tree (see \cref{fig:polygon-decomp}).
Observe that as long as there is some diagonal \(D\) separating our query point
from a subset of points of \(A\), we can use the approach above.

\begin{figure}
\centering
\includegraphics{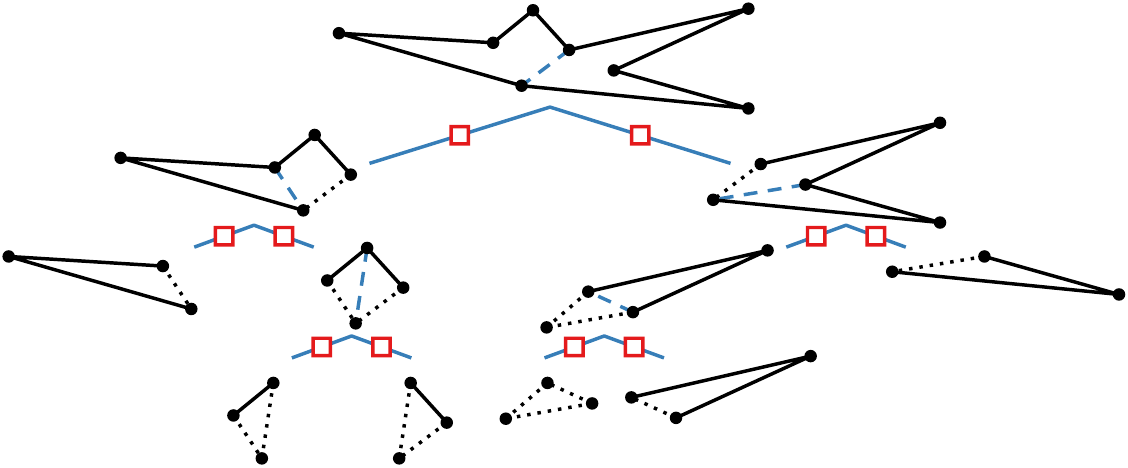}
\caption{Augmented polygon decomposition following the approach by
Chazelle~\cite{chazelle82}.
Each node corresponds to the splitting diagonal (blue dashed line).
Along the tree edges (blue lines), we store the multilevel cutting tree (red
box) for the polygon in the child using the diagonal of the parent.}
\label{fig:polygon-decomp}
\end{figure}

Every node of the tree is associated with a diagonal, and the two children
correspond to the left and the right subpolygon.
With each node, we store two data structures described above: one for the query
point to the left of the diagonal and one for the query to the right.

The query then proceeds as follows.
Suppose the polygon \(P\) is triangulated, and the triangles correspond to the
leaves in the decomposition.
Given a query point \(q\), find the triangle it belongs to; then traverse the
tree bottom up.
In the leaf, \(q\) can see all the points of \(A\) that are in the same
triangle, so we start with that count.
As we proceed up the tree, we query the correct associated data structure---if
\(q\) is to the right of the diagonal, we want to count the points to the left
of the diagonal in the current subpolygon that see \(q\).
It is easy to see that this way we end up with the total number of points in
\(A\) that see \(q\), since the subsets of \(A\) that we count are disjoint as
we move up the tree and cover the entire set \(A\).

\theoremPointPoint*
\begin{proof}
The correctness follows from the considerations above; it remains to analyse the
time and storage requirements.
For the query time, we do point location of the query point \(q\) in the
triangulation of \(P\) and make a single pass up the decomposition tree, making
queries in the associated multilevel cutting trees.
Clearly, the height of the tree is \(\bO(\log n)\).
At every level of the decomposition tree, we need to construct the visibility
cone from the query point to the current diagonal; this can be done in
\(\bO(\log n)\) time using the shortest path data structure for
\(P\)~\cite{guibas89,hershberger91}, as discussed in \cref{sec:prelims}.
Then we need to query the associated data structure, except at the leaf, where
we simply fetch the count.
The query then takes time
\[\sum_{i = 0}^{\bO(\log n)} \bO(\log m_i + \log n)
= \bO(\log n \log m + \log^2 n)\,,\]
where the sum of all \(m_i\) is at most \(m\).
For the storage requirements, we need to store the associated data structures on
in total \(m\) points at every level of the tree, as well as a single copy of
the shortest path data structure, yielding overall
\(\bO(n + m^{2 + \eps} \log n)\) storage.
Finally, we analyse the preprocessing time.
Triangulating a simple polygon takes \(\bO(n)\) time~\cite{chazelle91}.
Constructing the decomposition can be done in additional \(\bO(n)\)
time~\cite{guibas89}.
Constructing the associated data structures takes time \(\bO(m^{2 + \eps})\) per
level, so \(\bO(m^{2 + \eps} \log n)\) overall, after determining the visibility
cones for the points of \(A\) to all the relevant diagonals, which can be done
in time \(\bO(m \log^2 n)\), as each point of \(A\) occurs a constant number of
times per level of the decomposition, and constructing the cone takes
\(\bO(\log n)\) time.
Overall we need \(\bO(n + m^{2 + \eps} \log n + m \log^2 n)\) time.
\end{proof}

\begin{figure}
\centering
\includegraphics{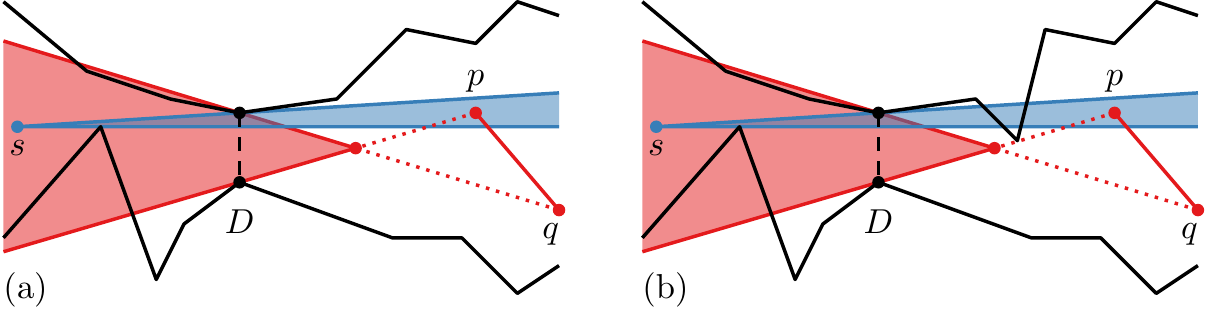}
\caption{(a)~For the cone that describes visibility of \(\seg{pq}\) through
\(D\), \cref{lem:cones} does not hold---there can be visibility without the
apices of the cones seeing each other.
(b)~The segment \(\seg{pq}\) intersects the cone of \(s\), and \(s\) is in the
cone of \(\seg{pq}\), but they cannot see each other, so testing intersection
between the objects and the cones also does not work directly.}
\label{fig:cone-lemma-fails-for-segments}
\end{figure}

\begin{remark}\label{rem:no_reuse}
While this approach uses many of the ideas needed to tackle the setting with
segment queries, \cref{lem:cones} does not apply---see
\cref{fig:cone-lemma-fails-for-segments}.
\end{remark}

\section{Segment Queries}\label{sec:point_segment}
In this \lcnamecref{sec:point_segment}, given a simple polygon \(P\) and a set
\(A\) of stationary entities (points) in \(P\), we construct a data structure
that efficiently counts the points in \(A\) visible from a query segment
\(\queryseg\).
We cannot reuse the approach of \cref{sec:arrangement}, since \(\queryseg\) may
intersect multiple arrangement cells, so the query time would depend on the
number of visible entities, as we need to keep the list to avoid double
counting; even if this issue were solved, we would need to sum up the values
from the \(\Omega(n)\) cells we might cross.
Therefore, we construct a new data structure using the insights of the
hierarchical decomposition of \cref{sec:point_point}.
That approach is also not directly usable, as discussed in \cref{rem:no_reuse}.

We use the data structure by Guibas and Hershberger~\cite{guibas89} (abbreviated
\spds) on \(P\), discussed in \cref{sec:prelims}, as the foundation.
For a given query \(\seg{pq}\), the \spds\ partitions \(P\) into four types of
regions (\cref{fig:dead_monkeys_alive_query_overview}): hourglasses (orange);
triangles that contain \(p\) or \(q\), denoted by \(T_L\) and \(T_R\) (blue);
regions that have as a border the upper or the lower chain of an hourglass,
referred to as \emph{side polygons} (green); and regions that have as a border
one of the edges of \(T_L\) or \(T_R\), referred to as \emph{end polygons}
(red).
The number of visible objects in \(A\) is the sum of the counts of objects in at
least one of the relevant hourglasses or triangles, plus the size of the set of
objects contained in a side or an end polygon that are visible to \(\seg{pq}\).
This allows us to subdivide the problem into tractable parts with strong
assumptions.

Counting the visible objects inside the relevant hourglasses and triangles is
easy, since all of them are visible.
For the side polygons, we make an observation regarding the conditions for an
object \emph{not} to be visible, and we subtract that count from the overall
count of points in the relevant side polygons.
Finally, for the end polygons, we make a case distinction on the way the
visibility cones of the objects cross the adjacent triangles, and use
inclusion--exclusion-style arguments to obtain the correct count.

\begin{figure}[tb]
\centering
\includegraphics{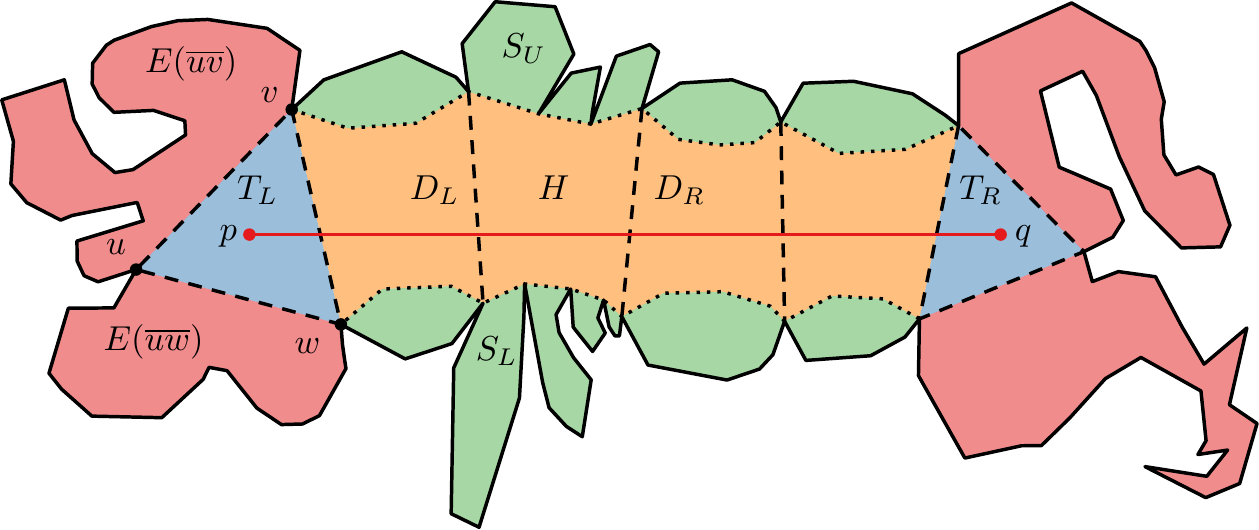}
\caption{Partitioning of the polygon based on the polygon cover of
\(\seg{pq}\).}
\label{fig:dead_monkeys_alive_query_overview}
\end{figure}

\subsection{The Data Structure}\label{sec:ps_ds}
In this \lcnamecref{sec:ps_ds}, we describe what our data structure stores and
how to compute it.
We start with some helpful observations, leading to a helper data structure.

\begin{lemma}\label{lem:ds_viscone}
Let \(H\) be an hourglass bounding a side polygon \(S\).
Denote the left diagonal of \(H\) by \(D_L\); and denote the polygon bounded by
\(D_L\) by \(P_L\).
Let \(\Rays\) be a set of visibility rays from objects in \(S\) into \(H\) that
exit \(H\) through \(D_L\).
(See \cref{fig:visibility-proof}.)
In time \(\bO(\SetSize{\Rays}^{2 + \eps})\), we can compute a data structure of
size \(\bO(\SetSize{\Rays}^{2 + \eps})\), so that given a query segment
\(\queryseg \subset P_L\), we can count the number of rays in \(\Rays\) that are
intersected by \(\queryseg\) in time \(\bO(\log \SetSize{\Rays})\).
\end{lemma}
\begin{proof}
Suppose that some ray \(R \in \Rays\) intersects \(\queryseg\).
Since \(R\) is a visibility ray into \(H\) that intersects \(D_L\), it extends
into \(P_L\) and its apex is to the right of the supporting line of \(D_L\).
Since \(\queryseg \subset P_L\), it follows that \(R\) can only intersect
\(\queryseg\) left of the supporting line of \(D_L\).

Let \(\seg{ps}\) be the subsegment of \(\queryseg\) that is left of this line.
We can compute \(\seg{ps}\) in constant time.
The segment \(\seg{ps}\) intersects \(R\) if and only if it intersects the
supporting line of \(R\).
Testing for the intersection between a preprocessed set of lines and a query
segment can be done with a two-level cutting tree
(\cref{lem:multilevel-cutting-tree-2}), concluding the proof.
\end{proof}

\begin{figure}
\centering
\includegraphics{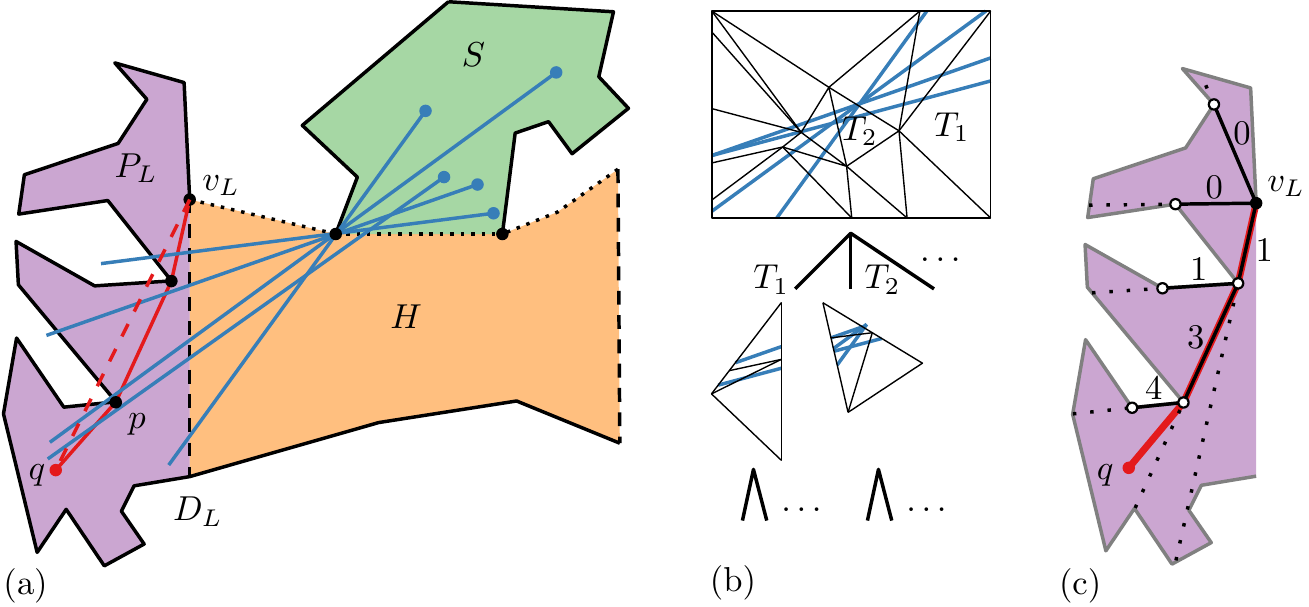}
\caption{(a)~We want to count the blue rays intersecting the shortest path
between \(q\) and \(v_L\).
(b)~We store a multilevel cutting tree to query ray intersections with
\(\seg{pq}\) and \(\seg{qv_L}\).
(c)~We store a shortest path map to count the rays intersecting the shortest
path from \(v_L\) to \(p\).
In this case, we count \(\sfrac{1}{2}\cdot (3 + 1 + 4) = 4\).}
\label{fig:visibility-proof}
\end{figure}

\begin{lemma}\label{lem:sp_intersection_ds}
Let \(H\) be an hourglass bounding a side polygon \(S\).
Denote the left diagonal of \(H\) by \(D_L\); and denote the polygon bounded by
\(D_L\) by \(P_L\).
Let \(\Rays\) be a set of visibility rays from objects in \(S\) into \(H\) that
exit \(H\) through \(D_L\).
Finally, denote the leftmost vertex of the convex chain separating \(H\) from
\(S\) by \(v_L\).
(See \cref{fig:visibility-proof}.)
Given a query point \(q \in P_L\) to the left of the supporting line of \(D_L\),
whose shortest path to \(v_L\) in \(P_L\) forms an upwards convex chain, we wish
to count the rays in \(\Rays\) that intersect this convex chain.
In time \(\bO(\SetSize{\Rays}^{2 + \eps} + \SetSize{P_L} \log\SetSize{\Rays})\),
we can compute a data structure of size
\(\bO(\SetSize{\Rays}^{2 + \eps} + \SetSize{P_L})\) that can answer such queries
in time \(\bO(\log\SetSize{\Rays}\SetSize{P_L})\).
\end{lemma}
\begin{proof}
The shortest path from \(q\) to \(v_L\), together with \(\seg{qv_L}\), forms a
convex polygon.
A ray starting in \(S\) (and thus outside this convex polygon) that intersects
an edge of the convex polygon must intersect exactly two edges of the polygon.

We store \(\Rays\) in the data structure of \cref{lem:ds_viscone}.
We also store a \emph{shortest path map} on \(P_L\) with \(v_L\) as its
root~\cite{guibas87}, computed in time \(\bO(\SetSize{P_L})\) and consisting of
\(\bO(\SetSize{P_L})\) segments.
For every edge in the map, we query the data structure of \cref{lem:ds_viscone}
to obtain the number of rays in \(\Rays\) that intersect that edge.
With every vertex \(s\) in the shortest path map, we store the total number of
intersections between the rays in \(\Rays\) and the path from \(v_L\) to \(s\).
Constructing this augmented data structure uses \(\bO(\SetSize{P_L})\) space and
takes time \(\bO(\SetSize{P_L} \log\SetSize{\Rays})\).

Given this data structure, we answer a query as follows.
That shortest path from \(q\) to \(v_L\) consists of a convex chain of segments
of the shortest path map between \(v_L\) and some vertex \(p\), followed by a
segment \(\seg{pq}\).
In time \(\bO(\log\SetSize{P_L})\), we can identify the vertex \(p\).
We query the vertex \(p\) for the (stored) total number of intersections with
the rays in \(\Rays\) in \(\bO(1)\) time.
Then we query the data structure of \cref{lem:ds_viscone} with segments
\(\seg{pq}\) and \(\seg{qv_L}\) in \(\bO(\log\SetSize{\Rays})\) time and add all
three counts together.
This way, we count all the intersections of the rays in \(\Rays\) with the
boundary of the convex polygon defined by \(\seg{qv_L}\) and the shortest path
from \(q\) to \(v_L\).
For each ray, we now count both intersections; so we divide the total by two and
return the result.
This procedure is needed, since a ray can intersect the shortest path two times
(and not intersect \(\seg{qv_L}\)).
The query takes time \(\bO(\log\SetSize{\Rays} + \log\SetSize{P_L}) =
\bO(\log\SetSize{\Rays}\SetSize{P_L})\).
\end{proof}

We now introduce our Segment Query Data Structure (SQDS).
It is based on the \spds, augmented with extra data for visibility queries.
The \spds\ decomposes the polygon into hourglasses and triangles; we describe
the data structures we store with each class.

\begin{figure}[tb]
\centering
\includegraphics{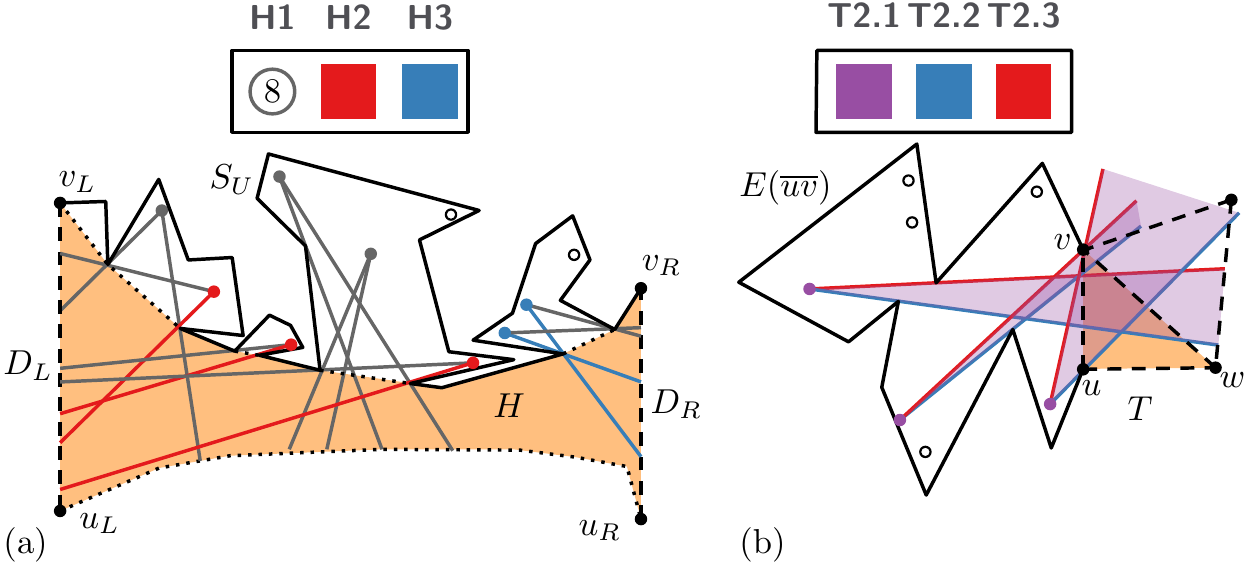}
\caption{(a)~Data structures per chain of an hourglass.
(b)~Data structures in \DSPartRef{T2}, for cones from \(E(\seg{uv})\) restricted
to pass through \(\seg{vw}\).}
\label{fig:sqds-hourglass-and-triangles}
\end{figure}

\subparagraph*{The data structure for hourglasses.}
Consider an hourglass \(H\) bounded by diagonals \(D_L = \seg{v_Lu_L}\),
\(D_R = \seg{v_Ru_R}\), and the upper and the lower chains \(\pi(v_L, v_R)\) and
\(\pi(u_L, u_R)\) in the \spds\ (\cref{fig:sqds-hourglass-and-triangles}a).
Let \(S_U\) be the side polygon of \(H\) that is incident to the upper chain,
and let \(\C_U\) be the visibility cones of entities in \(S_U\) into \(H\).
For the hourglass \(H\) itself, we store the number of objects in \(A\) that are
contained in \(H\).
For ease of exposition, we refer to the boundaries of a cone \(C \in \C_U\) as
the \emph{left} and the \emph{right} boundary, when viewed from the apex of the
cone in the direction of the cone.
For the upper chain of \(H\), we store in SQDS:
\begin{description}
\item[H1.] The number of non-empty visibility cones in \(\C_U\).
\item[H2.] The right cone boundaries of all cones in \(\C_U\) that fully exit
\(H\) through \(D_R\) in the data structure of \cref{lem:sp_intersection_ds}.
\item[H3.] The left cone boundaries of all cones in \(\C_U\) that fully exit
\(H\) through \(D_L\) in the data structure of \cref{lem:sp_intersection_ds}.
\end{description}
We store symmetrical data structures for the bottom chain of \(H\) (i.e.\@ the
left cone boundaries for the cones that exit through \(D_R\) and the right cone
boundaries for the cones that exit through \(D_L\)).

\subparagraph*{The data structure for triangles.}
Let \(T = uvw\) be a triangle in the polygon decomposition underlying the \spds.
For \(T\), we store the number of objects in \(A\) that are contained in \(T\).
In addition, consider an edge \(\seg{uv}\) of \(T\).
To introduce our data structure, we assume that \(\seg{uv}\) is vertical.
Denote the subpolygon (end polygon) adjacent to \(T\) and bounded by
\(\seg{uv}\) by \(E(\seg{uv})\); for ease of exposition, assume \(T\) is to the
right of \(E(\seg{uv})\).
Let \(\C_{E(\seg{uv})}\) be the visibility cones into \(T\) of objects in
\(A \cap E(\seg{uv})\).
Let \(\C_{E(\seg{uv})}(\seg{vw})\) be the set of subcones of the cones in
\(\C_{E(\seg{uv})}\) that intersect \(\seg{vw}\); similarly, define
\(\C_{E(\seg{uv})}(\seg{uw})\) as the subcones intersecting \(\seg{uw}\).
We store:
\begin{description}
\item[T1.] All cones \(\C_{E(\seg{uv})}\) in a data structure of
\cref{lem:separated-cones-seg-ds}.
\item[T2.] The subcones of \(\C_{E(\seg{uv})}(\seg{vw})\) in
(\DSPartRef{T2.1})~a data structure of \cref{lem:separated-cones-seg-ds} and
(\DSPartRef{T2.2})~the right and (\DSPartRef{T2.3})~the left boundaries of these
cones in data structures of \cref{lem:sp_intersection_ds}.
\item[T3.] Symmetric to \DSPartRef{T2}, but for \(\C_{E(\seg{uv})}(\seg{uw})\).
\item[T4.] All the left cone boundaries that intersect \(\seg{vw}\) in a cutting
tree.
\item[T5.] All the right cone boundaries that intersect \(\seg{uw}\) in a
cutting tree.
\end{description}
This list contains the data structures we store for \(\seg{uv}\); we store
analogous data structures for edges \(\seg{vw}\) and \(\seg{uw}\) of \(T\).
See also \cref{fig:sqds-hourglass-and-triangles}b.

\begin{lemma}
\label{lem:sqds-size-and-build-time}
The SQDS requires \(\bO(nm^{2 + \eps} + n^2)\) space and can be constructed in
time \(\bO(nm^{2 + \eps} + nm\log n + n^2\log m)\).
\end{lemma}
\begin{proof}
During preprocessing, we construct the \spds\ and explicitly store all the
hourglasses, requiring \(\bO(n\log n)\) space and time (\cref{sec:prelims}).
We augment the decomposition with our data structures.
There are \(\bO(n)\) hourglasses and \(\bO(n)\) triangles, so we need (by
\cref{lem:sp_intersection_ds}) at most \(\bO(nm^{2 + \eps} + n^2)\) space.

For each of our data structures, we need the visibility cones to the diagonals
of the triangles and hourglasses to compute the data structures.
We compute these in \(\bO(m\log n)\) per triangle or hourglass, giving a total
time of \(\bO(nm\log n)\).
Constructing our data structures, given the visibility cones, takes additionally
\(\bO(nm^{2 + \eps} + n^2\log m)\) time.
In total, we construct the segment query data structure in time
\(\bO(nm^{2 + \eps} + nm\log n + n^2\log m)\).
\end{proof}

Next, we describe how to query the data structure, so that we can efficiently
count visible entities.
Given a query segment \(\queryseg\), we compute the number of visible entities
in \(A\) from \(\queryseg\) using the SQDS.
We obtain the start triangle \(T\) containing \(p\), \(\bO(\log n)\)
hourglasses, and an end triangle containing \(q\).
Since all entities of \(A\) inside the hourglasses and the two triangles are
visible, we sum up their counts that we store in SQDS.
It remains to count the objects in the end polygons bounded by the triangles and
the two side polygons per hourglass.

\subsection{Counting Entities in the End Polygons}\label{sec:end}
Let \(T = uvw\) be a triangle in our decomposition that contains \(p\).
Assume without loss of generality that \(\seg{uv}\) is vertical with \(T\) to
its right and \(v\) above \(u\) (see \cref{fig:triangle_cases}a).
Let \(E(\seg{uv})\) be the end polygon adjacent to \(T\) left of \(\seg{uv}\),
and let \(\C_{E(\seg{uv})}\) be the set comprised, for each point
\(s \in A \cap E(\seg{uv})\), of visibility cones
\(V(s, \seg{uv}, E(\seg{uv})\).
To shorten the phrasing, when we say that a cone \(C \in \C_{E(\seg{uv})}\)
\emph{sees} a segment \(\queryseg\), we mean that there is mutual visibility
between the apex of that cone and \(\queryseg\) in the complete polygon \(P\).
A cone \(C \in \C_{E(\seg{uv})}\) sees \(\queryseg\) whenever \(\queryseg\)
intersects \(C\) and there are no edges of \(P \setminus E(\seg{uv})\) that
block visibility.
First, consider the special case where \(\queryseg\) is contained in \(T\).
Since there are no edges of \(P\) in \(T\), a cone \(C \in \C\) sees the segment
if and only if \(C\) intersects it, and we conclude:

\begin{observation}\label{obs:enclosed-segment-ds}
Using \DSPartRef{T1} stored per edge of the triangle \(T\), we can count the
objects in all the end polygons of \(T\) that see a query segment
\(\queryseg \subset T\) in \(\bO(\log m)\) time.
\end{observation}

\subsubsection{From Segments to Piercing Segments}\label{sec:non_pierce}
Now suppose that \(p \in T\) and \(q \notin T\).
We assume that \(\queryseg\) pierces \(\seg{vw}\) (the case when \(\queryseg\)
pierces \(\seg{uw}\) is symmetrical).
In this and the next \lcnamecrefs{sec:non_pierce}, we consider a case
distinction on the types of the cones, based on how they pass through \(T\) and
some adjacent triangle.
We argue that in all cases we can count the visible objects correctly.

We partition the set of cones \(\C\) into three classes using the vertex \(w\)
of \(T\) (\cref{fig:triangle_cases}):
\begin{itemize}
\item blue cones \(\cblue\) pass completely above \(w\);
\item red cones \(\cred\) pass completely below \(w\); and
\item purple cones \(\cpurple\) contain \(w\).
\end{itemize}
Next, we argue that we can count the number of visible cones per class.

\begin{figure}[tb]
\centering
\includegraphics{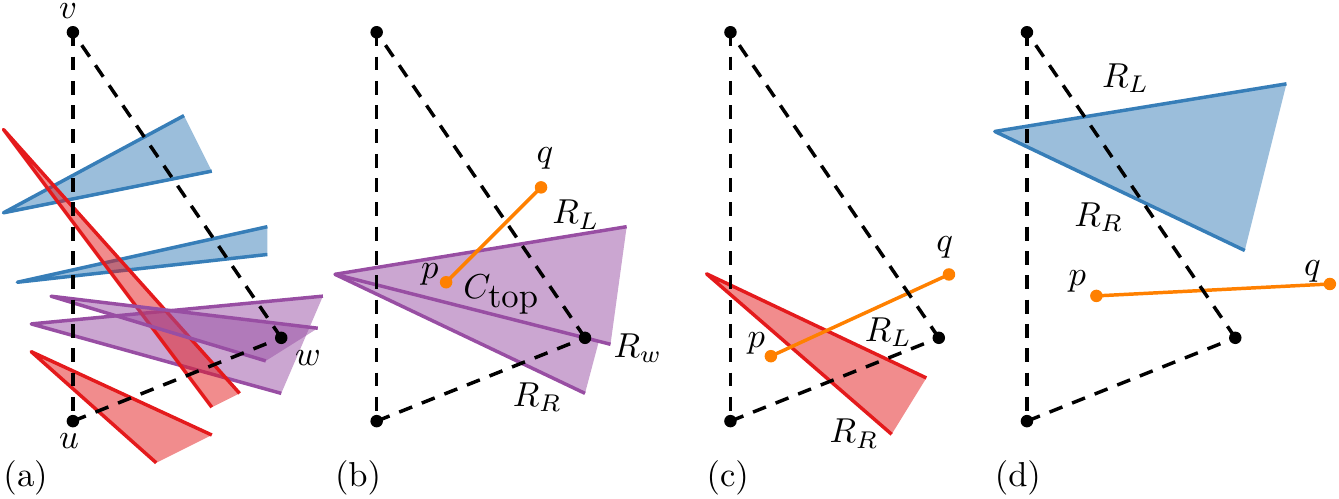}
\caption{(a)~The three classes of cones.
(b)~For the purple cones, it suffices to look at \(\conetop\).
(c)~For the red cones, we test \(R_R\).
(d)~For the blue cones, we need to look into the adjacent triangle.}
\label{fig:triangle_cases}
\end{figure}

\subparagraph*{Counting visible red cones.}
Consider a cone \(C \in \cred\).
Since \(C\) only intersects \(\seg{uv}\) and \(\seg{uw}\), and \(\queryseg\)
exits the triangle through \(\seg{vw}\) by assumption, any part of \(\queryseg\)
that \(C\) can see must be in \(T\) (see \cref{fig:triangle_cases}c).
Hence, \(C\) sees \(\queryseg\) if and only if \(p\) lies below the left ray of
\(C\).

\begin{lemma}\label{lem:count_red}
Using \DSPartRef{T4} (resp.\@ \DSPartRef{T5}), we can count the number of red
cones in \(\C_{E(\seg{uv})}\) that see \(\queryseg\), assuming \(\queryseg\)
pierces \(\seg{vw}\) (resp.\@ \(\seg{uw}\)), in \(\bO(\log m)\) time.
\end{lemma}
\begin{proof}
A red cone sees \(\queryseg\) if and only if it sees \(\seg{ps}\), which is
contained in \(T\).
Moreover, red cones originate to the left of the supporting line of
\(\seg{uv}\), and \(\seg{ps}\) is to the right of that line.
Thus, we can employ the cutting tree we stored in \DSPartRef{T4} to count the
cones that see \(\seg{ps}\).
\end{proof}

\subparagraph*{Counting visible purple cones.}
Consider a cone \(C \in \cpurple\); we show the following
\lcnamecref{lem:top_sees_all}.
\begin{lemma}\label{lem:top_sees_all}
Let \(C \in \cpurple\) be bounded by the right and the left rays \(R_R\),
\(R_L\), and let \(R_w\) be the ray from the apex of \(C\) through the vertex
\(w\).
The query segment \(\queryseg\) sees \(C\) if and only if it sees the cone
\(\conetop\) bounded by \(R_L\) and \(R_w\).
\end{lemma}
\begin{proof}
If \(\queryseg\) is visible to \(\conetop\), then it is visible to \(C\).
It remains to show that if \(\queryseg\) is visible to \(C\), then it must be
visible to \(\conetop\).

The point \(p\) is either contained in \(\conetop\), in the area of \(T\) above
\(\conetop\), or in the area of \(T\) below \(\conetop\).
If \(p\) is contained in the area below \(\conetop\), then, since \(\queryseg\)
pierces \(\seg{vw}\), the segment \(\queryseg\) always intersects the ray from
the apex of \(C\) through \(w\), and the segment is thus always visible to
\(\conetop\) (and \(C\)).
If \(p\) is contained in the area above \(\conetop\), then \(\queryseg\)
intersects \(C\) if and only if it enters \(C\) through the left ray \(R_L\).
If \(\queryseg\) intersects \(R_L\) in \(T\), then \(\queryseg\) and
\(\conetop\) (and \(C\)) are always mutually visible, since there are no polygon
vertices in \(T\) that may block the line of sight.
If \(\queryseg\) intersects \(R_L\) outside of \(T\), then \(\queryseg\) cannot
be visible to \(C \setminus \conetop\).
Finally, if \(p\) is contained in \(\conetop\), then the segment \(\queryseg\)
is always visible to \(\conetop\) (and \(C\)).
\end{proof}

\Cref{lem:top_sees_all} implies that whenever \(\queryseg\) intersects
\(\seg{vw}\), we can use \(\conetop \in \ctop\) instead of \(C \in \cpurple\),
where \(\ctop\) is the collection of all such cones.
These cones in \(\ctop\) match the definition of cones in \(\cblue\), and we
store them together in \DSPartRef{T2}.

\subparagraph*{Counting visible blue and top cones.}
Consider a cone \(C \in \ctop \cup \cblue\).
Since \(C\) may see \(\queryseg\) outside triangle \(T\), we consider the
triangle \(\tadj \neq T\) that is incident to \(\seg{vw}\).
If \(q\) is inside \(\tadj\), we can use the data structure of
\cref{lem:separated-cones-seg-ds} in \DSPartRef{T2} to count the number of
visible blue cones, since visibility cannot be blocked in \(\tadj\).
In the next \lcnamecref{sec:pierce}, we show that if \(q\) is outside \(\tadj\),
we can count the visible blue cones in time \(\bO(\log nm)\) using
\DSPartRef{T2} and \DSPartRef{T3}.

\subsubsection{Counting Visible Cones in an Adjacent Triangle}\label{sec:pierce}
\begin{figure}[tb]
\centering
\includegraphics{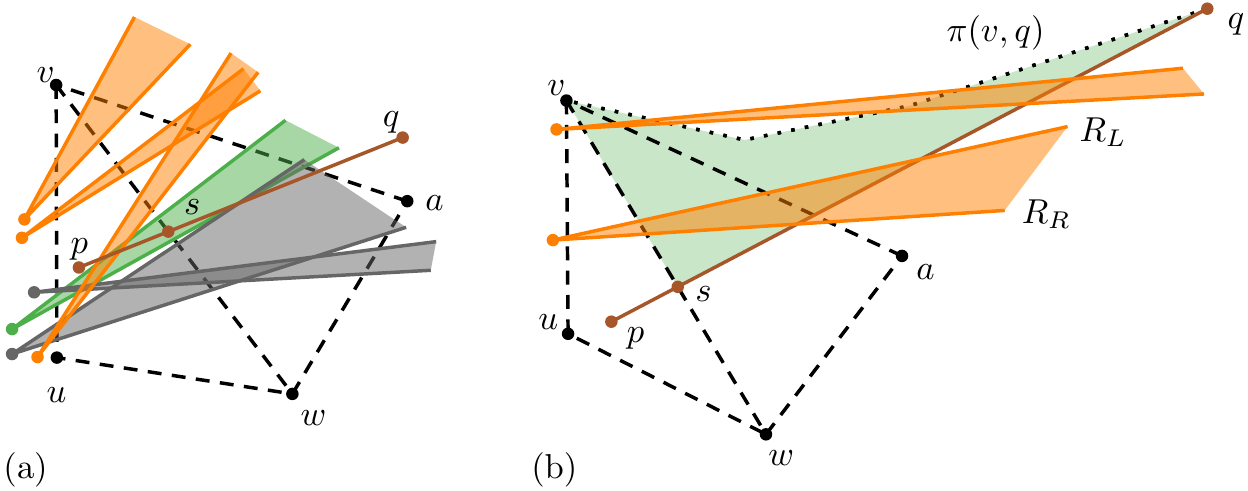}
\caption{(a)~The three types of cones based on the intersection point \(s\).
(b)~For the orange cones, we can test visibility along \(R_R\), and it can only
be blocked by the upper chain \(\pi(v, q)\) of \(F\).}
\label{fig:piercing_case}
\end{figure}

We assume that \(\queryseg\) pierces \(T\) in the edge \(\seg{vw}\) and then
passes through some triangle \(\tadj = wva\).
Let \(s\) be the intersection point between \(\seg{vw}\) and \(\queryseg\).

Abusing (or overloading) notation, we again partition the set of cones
\(\ctop \cup \cblue\) in three classes, but now using \(s\): the orange cones
\(\corange\) pass above \(s\), the grey cones \(\cgrey\) pass below \(s\), and
the green cones \(\cgreen\) contain \(s\) (\cref{fig:piercing_case}a).
Since there are no polygon vertices in \(\tadj\), all the cones in \(\cgreen\)
see \(\queryseg\).
Next, we consider all the cones in \(\corange\)---all the cones in \(\cgrey\)
are treated in a symmetrical fashion.

\begin{lemma}\label{lem:type_2_reduce_bottom}
A cone \(C \in \corange\) sees \(\queryseg\) if and only if it sees
\(\queryseg\) along the right cone boundary \(R_R\).
\end{lemma}
\begin{proof}
Let \(R\) be a ray in \(C\) such that \(C\) sees \(\queryseg\) along \(R\).
Let \(z\) be the intersection point of \(R\) with \(\queryseg\).
Since the apex of \(C \in \corange\) lies left of \(\seg{vw}\), the slope of
\(R\) and \(R_R\) must be lower than that of the supporting line of
\(\queryseg\).
By definition of the right cone boundary (and our chosen orientation of
\(\seg{vw}\)), the slope of \(R\) is at least the slope of \(R_R\).
Thus, if \(R\) intersects \(\queryseg\), then \(R_R\) intersects \(\queryseg\)
left of that point of intersection.
Since \(R\) realises mutual visibility between \(C\) and \(\queryseg\), the area
bounded by \(R\) and \(\queryseg\) cannot contain any polygon vertices, and thus
\(R_R\) must also realise mutual visibility between \(\queryseg\) and \(C\).
\end{proof}

Let \(\B_c\) be the right rays of the cones of class \(c\), e.g.\@ \(\borange\)
are the right rays of the cones in \(\corange\), and let \(\U_c\) be the left
rays of the cones of class \(c\).
Let \(N_c\) be the number of cones of class \(c\) that see \(\queryseg\).
Similarly, define \(N_c(\seg{xy})\) as the number of cones of class \(c\) that
see some segment \(\seg{xy}\).
For some two points \(x\) and \(y\) and a class \(c\), denote the number of rays
from \(\B_c\) along which we can see the segment \(\seg{xy}\) in \(P\) by
\(N^\B_c(\seg{xy})\); and similarly, use \(N^\U_c(\seg{xy})\) for \(\U_c\).

\begin{lemma}\label{lem:visible}
For the query segment \(\queryseg\), we have that \(\norange =
\norange(\seg{ps}) + \norange^\B(\seg{sq})\).
\end{lemma}
\begin{proof}
By \cref{lem:type_2_reduce_bottom}, an orange cone \(C\) can see \(\queryseg\)
if and only if \(\queryseg\) is visible along \(R_R \subseteq C\).
Hence, \(\norange = \norange^\B(\queryseg)\).
Now observe that a ray in \(C\) cannot intersect both \(\seg{ps}\) and
\(\seg{sq}\), since \(C\) is not a green cone.
Finally, by \cref{lem:type_2_reduce_bottom},
\(\norange(\seg{ps}) = \norange^\B(\seg{ps})\).
This implies the statement of the \lcnamecref{lem:visible}.
\end{proof}

Therefore, we can count visibility of orange cones separately for \(\seg{ps}\)
and \(\seg{sq}\).
Let \(F\) be the funnel from \(q\) to \(\seg{vw}\).

\begin{lemma}\label{lem:only_top_chain_can_block}
Let \(R\) be a ray in \(\borange\).
Visibility of \(\seg{sq}\) along \(R\) is blocked if and only if \(R\)
intersects the top boundary of the funnel \(F\) before intersecting
\(\seg{sq}\).
\end{lemma}
\begin{proof}
Note that \(\seg{pq}\) crosses \(\seg{vw}\), so \(\seg{sq} \subseteq F\).
Consider the area bounded by the upper chain of \(F\) and \(\seg{sq}\).
This area cannot contain any polygon edges or vertices.
(See \cref{fig:piercing_case}b.)
The point of intersection between \(R\) and \(\seg{sq}\) is contained within
this area, if it exists.
Moreover, the triangles \(\tadj\) and \(T\) immediately left of \(\seg{sq}\)
also cannot contain any polygon vertices.
Thus, either \(R\) realises mutual visibility, or \(R\) intersects a polygon
edge that belongs to the funnel \(F\).
\end{proof}

\begin{corollary}\label{cor:nb_orange}
\(\norange^\B(\seg{sq}) = \SetSize{\borange} - X\), where \(X\) is the number of
rays in \(\borange\) that intersect the upper boundary of the funnel.
\end{corollary}

Now that we have established the necessary relations, we want to bring them
together and present a way to count the blue and the top cones that see the
query segment.

\begin{lemma}\label{lem:pierce_triangle_ds}
Using \DSPartRef{T2}, we can compute the number of blue and top cones that see
the segment \(\queryseg\), assuming \(p \in T\) and \(q \notin T \cup \tadj\),
in time \(\bO(\log nm)\).
\end{lemma}
\begin{proof}
Given \(\queryseg\), we compute \(s\) in \(\bO(1)\) time.
We want to compute \(\ngreen + \ngrey + \norange\).
Using the fact that the observations for the orange and the grey cones are
symmetric, we can use \cref{lem:visible}.
For the grey cones, the left cone boundaries matter instead of the right cone
boundaries for the orange cones.
We want to find \(\ngreen + \ngrey(\seg{ps}) + \norange(\seg{ps}) +
\ngrey^\U(\seg{sq}) + \norange^\B(\seg{sq})\).

It is easier not to separate the classes when checking visibility for
\(\seg{ps}\).
Since \(\tadj\) and \(T\) contain no polygon vertices, we can find
\(\ngreen + \ngrey(\seg{ps}) + \norange(\seg{ps})\) by using the data structure
of \cref{lem:separated-cones-seg-ds} stored in \DSPartRef{T2} in \(\bO(\log m)\)
time by querying visibility of \(\seg{ps}\).

Finally, we determine \(\ngrey^\U(\seg{sq})\) and \(\norange^\B(\seg{sq})\).
We argue the query time for \(\norange^\B(\seg{sq})\); the query for
\(\ngrey^\U(\seg{sq})\) can be done symmetrically.

The proof is illustrated by \cref{fig:piercing_case}b.
We show that we can uniquely count all right rays that are not in
\(\norange^\B(\seg{sq})\).
Again, we do not want to explicitly classify the cones into green, orange, and
grey, so we need a way to filter out the green and grey cones, as well as the
orange cones that do not see \(\seg{sq}\).

For any right ray \(R_R \in \bgreen \cup \bgrey\), we know that the point \(s\)
is \emph{above} the supporting line of \(R_R\) and either
\begin{itemize}
    \item \(R_R\) intersects \(F\), or
    \item \(\seg{sq}\) is above the supporting line of \(R_R\).
\end{itemize}
For any right ray \(R_R \in \borange\), we know, by
\cref{lem:only_top_chain_can_block,cor:nb_orange}, that \(s\) is \emph{below}
the supporting line of \(R_R\) and either
\begin{itemize}
    \item \(R_R\) intersects \(F\), or
    \item the orange cone sees \(\seg{sq}\) along \(R_R\).
\end{itemize}

It immediately follows that \(\norange^\B(\seg{sq})\) is equal to the total
number of rays \(\SetSize{\bgreen \cup \bgrey \cup \borange}\) minus the number
of right rays whose supporting line is below \(\seg{sq}\) and the number of
right rays that intersect \(F\).
We store the total count of the cones in \DSPartRef{T2}.
We can compute the first count we subtract by reusing one part of the data
structure of \cref{lem:separated-cones-seg-ds} we store in \DSPartRef{T2}.
We can compute the second count using the data structure of
\cref{lem:sp_intersection_ds} in \DSPartRef{T2}.
Together, these queries require \(\bO(\log m + \log n) = \bO(\log nm)\) time.
\end{proof}

Now that we have a procedure to count the blue and the top cones, we want to
bring our results together to count all visible entities in end polygons.
We now combine \cref{lem:count_red,lem:top_sees_all,lem:pierce_triangle_ds} into
the following \lcnamecref{lem:triangle_count_ds}.

\begin{lemma}\label{lem:triangle_count_ds}
Using \DSPartRef{T1--T3} of our SQDS, we can, given a query segment
\(\seg{pq}\) with \(p \in T\), and funnels from \(q\) to the edges of \(T\),
count the number of visible entities in the end polygons of \(T\) in time
\(\bO(\log nm)\).
\end{lemma}
\begin{proof}
With the query segment and the given funnels, we can determine which edges of
\(T\) are incident to end polygons and thus should be queried.
By \cref{lem:top_sees_all}, we can count the purple cones as the blue cones,
and they are already stored in \DSPartRef{T2}.
The red and the blue (or top) cones are queried separately, using
\DSPartRef{T4--T5} and \DSPartRef{T2--T3}, respectively.
Using \cref{lem:count_red,lem:pierce_triangle_ds}, we then query the data
structures associated with the edges of \(T\) to count the cones of the
different classes for each end polygon.
This requires \(\bO(\log m)\) and \(\bO(\log nm)\) time for the different
classes per edge, thus giving a total query time of \(\bO(\log nm)\).
\end{proof}

\subsection{Counting Entities in Side Polygons}
We now describe how to count entities in the side polygons of the hourglasses.
Let \(H\) be an hourglass that covers a part of query segment \(\queryseg\) (see
\cref{fig:hourglass-visibility}).

\begin{figure}[t]
\centering
\includegraphics{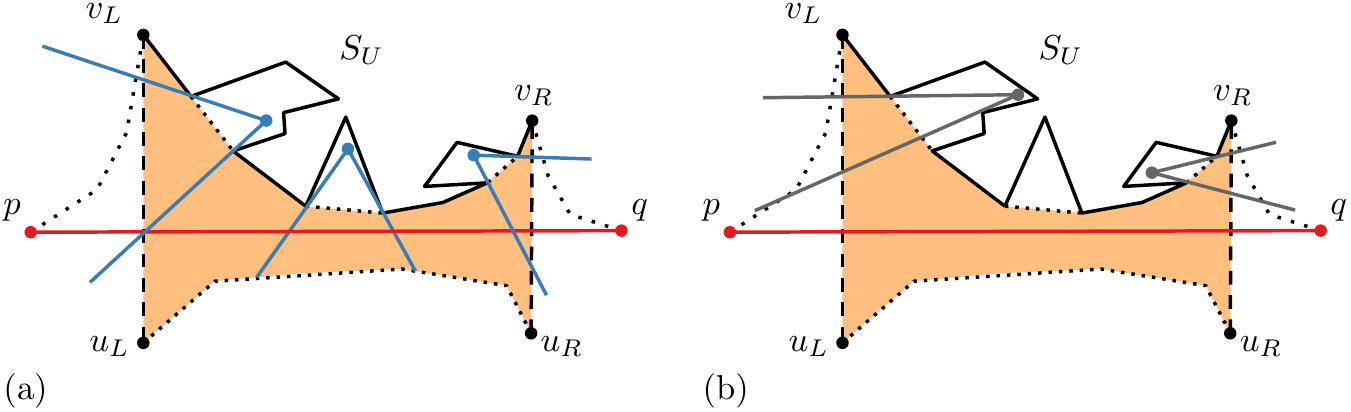}
\caption{Cones entering from side polygon \(S_U\) that (a)~see or (b)~do not see
\(\queryseg\).}
\label{fig:hourglass-visibility}
\end{figure}

\begin{lemma}\label{lem:hourglassvisibility}
Let \(H\) be an hourglass with diagonals \(D_L = \seg{u_Lv_L}\) and
\(D_R = \seg{u_Rv_R}\), let \(S_U\) be the side polygon bounded by the upper
chain of \(H\), and let \(\seg{pq}\) be a segment that intersects both \(D_L\)
and \(D_R\), with \(p\) to the left of \(D_L\) and \(q\) to the right of
\(D_R\).
Let \(a \in S_U\) be a point with a non-empty visibility cone \(C\) into \(H\).
Then point \(a\) does not see \(\queryseg\) if and only if either:
\begin{itemize}
    \item the right boundary \(R_R\) of \(C\) intersects \(\pi(v_R, q)\), or
    \item the left boundary \(R_L\) of \(C\) intersects \(\pi(v_L, p)\).
\end{itemize}
\end{lemma}
\begin{proof}
First, assume that \(a\) does not see \(\queryseg\).
We argue that \(R_R\) intersects \(\pi(v_R, q)\) or \(R_L\) intersects
\(\pi(v_L, p)\).
Assume for the sake of contradiction that neither condition holds.
Let \(I\) be the region bounded by \(\seg{v_Lv_R}\), \(\pi(v_R, q)\),
\(\queryseg\), and \(\pi(p, v_L)\).
Since \(C\) can see points in \(H\) along \(R_R\) and \(R_L\), \(R_R\) and
\(R_L\) enter the region \(I\) through \(\seg{v_Lv_R}\), or \(a\) already lies
inside \(I\).
Since \(R_R\) is a ray, it must also exit \(I\), and by definition it cannot
exit through \(\seg{v_Lv_R}\).
It cannot exit \(I\) through \(\queryseg\), either, as that would mean \(a\) can
see \(\queryseg\).
Furthermore, by our assumption, \(R_R\) also does not intersect \(\pi(v_R, q)\).
Hence, \(R_R\) intersects \(\pi(v_L, p)\).
Using an analogous argument, \(R_L\) must intersect \(\pi(v_R, q)\).
However, it now follows that the intersection point \(s = \queryseg \cap D_L\)
lies inside the cone \(C\), and must therefore be visible to \(a\) (i.e.\@
nothing above \(H\) can intersect \(\seg{ap}\), and inside \(H\) \(\seg{ap}\)
also does not intersect any polygon vertices).
Hence, \(a\) sees \(\queryseg\).
Contradiction.

Now assume that \(R_R\) intersects \(\pi(v_R, q)\) (the case that \(R_L\)
intersects \(\pi(v_L, p)\) is symmetric).
We now argue that \(a\) cannot see \(\queryseg\).
Let \(I_R\) be the region bounded by \(\queryseg\), \(D_R\), and
\(\pi(v_R, q)\).
A point \(s\) on \(\queryseg\) is visible via a ray \(R\), entering via \(D_R\),
if it first exits the region \(I_R\) via \(\queryseg\).
Since \(R_R\) is not obstructed in \(H\), it must enter \(I_R\) via \(D_R\).
In addition, by assumption, it first exits via \(\pi(v_R, q)\).
If \(R_R\) intersects \(\pi(v_R, q)\) once, then by convexity of
\(\pi(v_R, q)\), it follows that \(q\) is below \(R_R\) and thus below any ray
\(R\) in the cone \(C\), thus it is not visible.
If \(R_R\) intersects \(\pi(v_R, q)\) twice, there is a subsegment of
\(\queryseg\) above \(R_R\).
The ray \(R_R\) now partitions \(I_R\) into three regions: one below \(R_R\),
containing points that cannot be visible, and two regions above \(R_R\).
The right region contains the subsegment of \(\queryseg\) that is still above
\(R_R\).
Consider now any ray \(R\) that could be a visibility ray to a point
\(x \in \queryseg\).
This ray must be above \(R_R\) and must intersect \(\queryseg\) at \(x\).
This means that it must traverse the region \(I_R\) from \(D_R\) to \(x\).
But since \(R\) must be above \(R_R\), it follows that it must cross the two
regions above \(R_R\), which are separated by a polygon boundary.
Thus, no \(x \in \queryseg\) is visible from \(a\).
\end{proof}

\begin{lemma}\label{lem:hourglass-visibility-ds}
Using \DSPartRef{H1}, \DSPartRef{H2}, and \DSPartRef{H3} stored with each chain
of hourglass \(H\) in our SQDS, we can count the visible objects in the side
polygons of \(H\) in time \(\bO(\log nm)\).
\end{lemma}
\begin{proof}
By \cref{lem:hourglassvisibility}, we can count the number of visible objects
from the upper side polygon by taking the total number of objects with non-empty
visibility cones from the upper side polygon and subtracting those for which
either \(R_R\) intersects \(\pi(v_R, q)\) or \(R_L\) intersects \(\pi(v_L, p)\).
Counting for the lower side polygon is symmetrical.

We store the number of entities with non-empty visibility cones from the side
polygon in \DSPartRef{H1}.
Then, we query our \DSPartRef{H2} and \DSPartRef{H3} data structures to count
the number of visibility cones that exit through \(D_L\) and \(D_R\) that do not
see \(\queryseg\).
This requires \(\bO(\log m + \log n) = \bO(\log nm)\) query time per chain,
yielding the total time.
\end{proof}

We are now ready to bring all results together and prove
\cref{thm:point_segment}.

\theoremPointSegment*
\begin{proof}
The size and preprocessing time follow from \cref{lem:sqds-size-and-build-time}.
For the query, we first acquire the polygon cover of the query segment via the
\spds\ in \(\bO(\log n)\) time (\cref{sec:prelims}).
We sum up the number of entities contained in the \(\bO(\log n)\) hourglasses
and \(\bO(1)\) triangles in \(\bO(\log n)\) time.
In case the query segment is contained in a single triangle \(T\) of the polygon
decomposition, we use \DSPartRef{T1} to add the entities visible in all end
polygons of \(T\) in \(\bO(\log m)\) time (\cref{obs:enclosed-segment-ds}).
Otherwise, for each of the hourglasses, we query the associated data structure
of \cref{lem:hourglass-visibility-ds} to obtain the visible objects in the side
polygons of the hourglass in time \(\bO(\log nm)\) per hourglass, which
contributes \(\bO(\log n \log nm)\) to the query time.
We compute the funnels to the edges of the triangles in \(\bO(\log^2 n)\) time
and query the triangle data structures in \(\bO(\log nm)\) time, which is
dominated by the hourglass query time.
\end{proof}

\section{Segment Query for a Set of Segments}\label{sec:segment_segment}
As a natural extension to the previous data structure, we now consider the
problem where we have a set \(A\) of non-degenerate segments and want to
determine the number of visible segments from query segment \(\queryseg\).
As we show next, we can reuse the approach of the previous
\lcnamecref{sec:point_segment} with some minor additions and answer this query
in polylogarithmic time.

A difficulty that arises in this setting is that the entities in \(A\) are no
longer partitioned by the polygon cover of \(\queryseg\), that is, segments in
\(A\) may start or end in the polygon cover or pass through the cover.
To be able to correctly count the visible cones, we propose instead to count the
cones that we \emph{cannot} see and subtract this from the total count.

\subsection{Using Visibility Glasses}
To compute visibility of the segments in \(A\), we use \emph{visibility glasses}
(\cref{sec:prelims}) for our segments.
Let \(\seg{ac}\) be a segment to the left of a (vertical) diagonal
\(D = \seg{uv}\) in \(P\) (see \cref{fig:visibility-glass}).
We now check what \(\seg{ac}\) sees in the subpolygon to the right of the
diagonal.
To do this, we construct the visibility glass \(L(\seg{ac}, D)\) between
\(\seg{ac}\) and the diagonal.
Eades et al.~\cite{eades20} show that \(L(\seg{ac}, D)\) is an hourglass defined
by some subsegment \(\seg{or} \subseteq \seg{ac}\) and a subsegment
\(\seg{wx} \subseteq D\) (potentially, \(o = r\) or \(w = x\)).
We can now compute the lines connecting the opposite endpoints of the visibility
glass that still provide visibility, that is, the lines through \(\seg{ox}\) and
\(\seg{rw}\).
Note that these lines define the most extreme slopes under which there can still
be visibility.
These two lines intersect in a single point \(i\).
We now consider this point and the lines through it as a new cone that describes
the visible region to the right of diagonal \(D\).
We call this cone the \emph{visibility glass cone.}
Note that the left and the right rays of the cone are actual visibility rays to
points \(o\) and \(r\) on the line segment for points to the right of \(D\).

\begin{figure}[tb]
\centering
\includegraphics{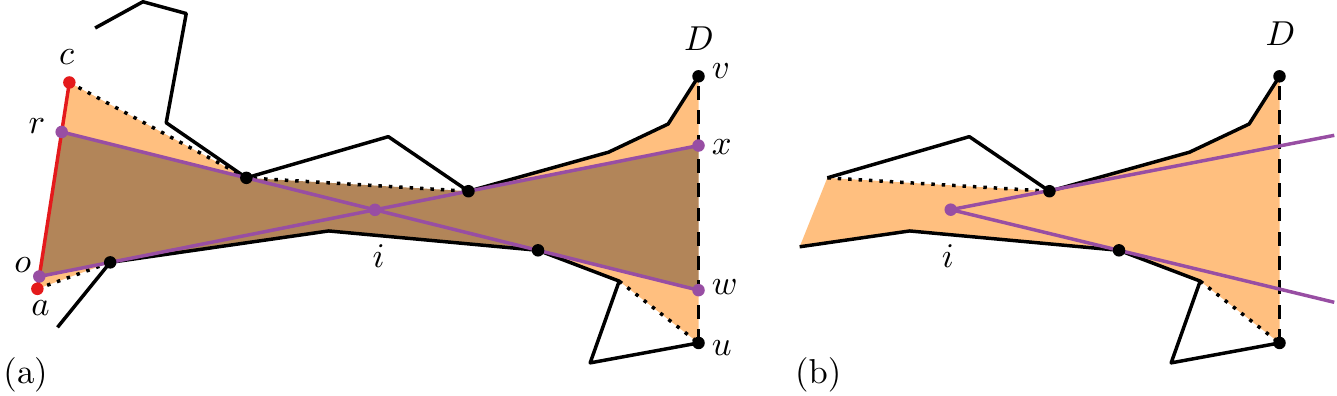}
\caption{(a)~The visibility glass (dark region) inside the hourglass (orange
region) from segment \(\seg{ac}\) to diagonal \(D = \seg{uv}\).
(b)~The intersection point \(i\) and the two rays from \(i\) through \(w\) and
\(x\) form a new visibility region in the subpolygon to the right of \(D\).}
\label{fig:visibility-glass}
\end{figure}

\begin{lemma}\label{lem:vis_glass_cone}
Consider a polygon \(P\), split into subpolygons \(P_L\) and \(P_R\) by a
diagonal \(D = \seg{uv}\), and let \(\seg{ac}\) be a line segment in \(P_L\).
Let \(C\) be the visibility glass cone of \(\seg{ac}\) into \(P_R\) through
\(D\).
If some point \(p \in P_R\) sees \(\seg{ac}\), it must be in \(C\).
\end{lemma}
\begin{proof}
Assume w.l.o.g.\@ that \(v\) is above \(u\), and let \(\seg{wx}\), with \(x\)
above \(w\), be the part of \(D\) inside \(C\).
See \cref{fig:visibility-glass}.
Suppose for a contradiction that \(p\) sees \(\seg{ac}\) but is not in \(C\).
Let \(q\) be a visible point on \(\seg{ac}\) and let \(L\) be the line segment
connecting \(p\) and \(q\).
Since \(D\) separates \(P_L\) from \(P_R\), and \(L\) must be inside \(P\) to be
a visibility line, \(L\) must cross \(D\).
Suppose w.l.o.g.\@ that \(L\) crosses \(D\) above \(x\), that is, above the left
ray \(R_L\) of \(C\).
The left ray must intersect a reflex vertex of the upper chain of its associated
hourglass \(H\), so there is a region above \(R_L\) bounded by the upper chain,
\(R_L\), and \(\seg{xv}\).
Since \(L\) enters this region, it must also exit the region.
There can only be visibility if \(L\) is inside \(P\), hence, it must exit the
region via the edge bounded by \(R_L\).
Therefore, its slope is higher than the slope of \(R_L\).
By definition of the visibility glass, it then cannot see \(\seg{ac}\), leading
to a contradiction.
Thus, \(p\) is in \(C\).
\end{proof}

\begin{corollary}
If \(p\) is visible, the ray from \(p\) through the apex \(i\) of \(C\) is a
visibility line to \(\seg{ac}\).
\end{corollary}

\Cref{lem:vis_glass_cone} shows that the visibility glass cones are functionally
the same as the visibility cones of points, thus we can reuse parts of our data
structures of \cref{sec:point_segment}.

\begin{observation}\label{obs:seg_vis_seg}
If a segment \(\seg{ac} \in A\) cannot see \(\queryseg\), it must be fully
contained in a side or an end polygon.
\end{observation}

This follows easily from the fact that if a segment \(\seg{ac}\) is not
contained in either an end or a side polygon, then it is in the polygon cover or
intersects the boundary of the polygon cover.
This then means that \(\seg{ac}\) sees \(\queryseg\).
It now suffices to count the segments in the side and the end polygons that are
not visible to determine the total number of entities invisible to
\(\queryseg\), and thus determine the number of entities in \(A\) visible to
\(\queryseg\).

Since our data structure of \cref{sec:point_segment} can already correctly count
visible entities from the end and the side polygons, we can simply determine the
number of invisible entities by subtracting the visible count from the total
number of entities in the end or side polygon.

\subsection{Extended Segment Query Data Structure}
We base our new data structure on the SQDS presented in
\cref{sec:point_segment}.
In our new data structure, for the upper chain of an hourglass \(H\) with the
incident side polygon \(S_U\), we store:
\begin{description}
\item[H1--3.] The same data structures as for the SQDS, but constructed using
visibility glass cones on line segments that are inside \(S_U\).
\item[H4.] The number of segments inside \(S_U\) with empty visibility glass
cones into \(H\).
\end{description}
We store symmetrical structures for the lower chain.

For a triangle \(T = uvw\) in the polygon decomposition, for the edge
\(\seg{uv}\) with the incident end polygon \(E(\seg{uv})\), we store:
\begin{description}
\item[T1--T5.] The same data structures as for the SQDS, but constructed using
visibility glass cones on line segments that are inside \(E(\seg{uv})\).
\item[T6.] The number of segments inside \(E(\seg{uv})\) with empty visibility
glass cones into \(H\).
\item[T7.] The number of segments inside \(E(\seg{uv})\) with non-empty
visibility glass cones into \(H\).
\end{description}

We now query the data structure as follows: for a given \(\queryseg\), we
acquire the polygon cover in \(\bO(\log n)\) time, giving us \(\bO(\log n)\)
hourglasses and bridges and \(\bO(1)\) end triangles.
For each side and end polygon, we now query the associated data structures to
get the number of visible segments for each subpolygon.
Let \(N_\textnormal{vis}\) be the total number of visible segments as reported
for the side and end polygons.
Over all the encountered side and end polygons during the query, let
\(N_\textnormal{closed}\) be the total number of empty visibility glass cones,
and let \(N_\textnormal{open}\) be the total number of non-empty visibility
glass cones.
The total number of visible segments is now given by
\(m - (N_\textnormal{open} - N_\textnormal{vis}) - N_\textnormal{closed}\).

\theoremSegmentSegment*
\begin{proof}
Since we only store a constant amount of extra data in our data structure per
hourglass chain and triangle edge, the storage requirements are the same as for
the SQDS.
For construction of the new data structure, we need to compute the visibility
glass cones for each side polygon of an hourglass and for all end polygons of a
triangle.
Using the data structure by Eades et al.~\cite{eades20}, we can compute the
visibility glasses to all diagonals in \(\bO(nm\log n + n \log^5 n)\) time and
extract the visibility glass cones in constant time per visibility glass.
Asymptotically, this does not change the preprocessing time.
Since the query only does a constant number of extra operations per hourglass
and triangle, the query time is the same as the original SQDS query time.
\end{proof}

\section{Extensions and Future Work}\label{sec:extra}
In this \lcnamecref{sec:extra}, we present some natural extensions to our work
and discuss possible variations, as well as the obstacles in the way of
obtaining results in those settings.

\subsection{Subquadratic Counting}\label{sec:count}
Given our data structures, we can generalise the problem: given two sets of
points or line segments \(A\) and \(B\), each of size \(m\), in a simple
polygon \(P\) with \(n\) vertices, count the number of pairs in \(A \times B\)
that see each other.
Using the work by Eades et al.~\cite{eades20} and further work it is based on,
we can solve this problem by checking the visibility for all pairs.
If \(n \gg m\), this approach is optimal.
In particular, if both sets \(A\) and \(B\) consist of points, this yields a
solution with running time \(\bO(n + m^2\log n)\); if one of the sets contains
only segments, we need \(\bO(n\log n + m^2\log n)\) time.
However, when \(m \gg n\), we want to avoid the \(m^2\) factor.
Furthermore, the setting of \cref{sec:segment_segment} is novel, so we consider
the full spectrum of trade-offs.

The trick is to use the following well-known technique.
Suppose we have a data structure for visibility counting queries with query
time \(Q(m, n)\) and preprocessing time \(P(m, n)\).
Pick \(k = m^s\) with \(0 \leq s \leq 1\).
We split the set \(A\) into sets \(A_1, \dots, A_k\), with \(\sfrac{m}{k}\)
objects each; then we construct a data structure for each set.
Finally, with each point in \(B\), we query these \(k\) data structures and sum
up the counts.
It is easy to see that the count is correct; the time that this approach takes
is \(\bO\bigl(k \cdot P(\sfrac{m}{k}, n) + mk \cdot Q(\sfrac{m}{k}, n)\bigr)\).
We need to pick \(s\) to minimise
\(\bO\bigl(m^s \cdot P(m^{1 - s}, n) + m^{1 + s} \cdot Q(m^{1 - s}, n)\bigr)\).

Let us show the results for the various settings.
Suppose that both sets \(A\) and \(B\) contain points.
First, let us consider the approach of \cref{sec:point_point}.
We have \(P(m, n) = \bO(n + m^{2 + \eps}\log n + m\log^2 n)\) and
\(Q(m, n) = \bO(\log^2 n + \log n \log m)\).
The summands depending only on \(n\) come from preprocessing of the polygon that
only needs to be done once; so we get
\begin{align*}
&\mathrel{\hphantom{=}}\bO\bigl(n + m^s \cdot (m^{(1 - s)(2 + \eps)}\log n
+ m^{1 - s}\log^2 n) + m^{1 + s} \log^2 n
+ m^{1 + s}\log n \log m^{1 - s}\bigr)\\
&= \bO(n + m^{(1 - s)(2 + \eps) + s}\log n + m^{1 + s} \log^2 n
+ m^{1 + s} \log n \log m)\,.
\end{align*}
Unless \(n \gg m\), we pick \(s\) such that \((1 - s)(2 + \eps) + s = 1 + s\);
we find \(s = \sfrac{(1 + \eps)}{(2 + \eps)}\).
Therefore, the running time is
\(\bO(n + m^{\sfrac{3}{2} + \eps'} \log n \log nm)\) for this choice of \(s\),
where \(\eps' > 0\) is an arbitrarily small constant.

Alternatively, we could apply the arrangement-based method of
\cref{sec:arrangement}.
We have \(P(m, n) \in \bO(nm^2 + nm \log n)\) and \(Q(m, n) \in \bO(\log nm)\).
Using the formula above, we get
\[\bO\bigl(m^s \cdot (nm^{2 - 2s} + nm^{1 - s} \log n)
+ m^{1 + s} \cdot \log(nm^{1 - s})\bigr)
= \bO(nm^{2 - s} + nm\log n + m^{1 + s}\log nm)\,.\]
If \(m \gg n\), we can pick \(s\) to balance the powers of \(m\) in the terms;
so we set \(s = \sfrac{1}{2}\) to get
\[\bO(m^{\sfrac{3}{2}}\cdot(n + \log n + \log m) + nm\log n)
= \bO(nm^{\sfrac{3}{2}} + m^{\sfrac{3}{2}}\log m + nm\log n)\,.\]

If \(n \gg m\), it is best to use the pairwise testing approach; however, if
\(m \gg n\), the arrangement-based approach performs best, and if
\(m \approx n\), we obtain best results with the decomposition-based approach of
\cref{sec:point_point}.

Now suppose that one of the sets contains points and the other set contains line
segments.
As it turns out, using the approach of \cref{sec:point_segment} is always
inefficient here; if \(m \gg n\), we can use the approach of
\cref{sec:arrangement}, making sure that we do point queries, and otherwise
pairwise testing is fastest.

Finally, suppose both sets consist of line segments.
We have \(P(m, n) \in \bO(n^2\log m + nm^{2 + \eps} + nm\log n)\) and
\(Q(m, n) \in \bO(\log^2 n + \log n \log m)\).
We get
\begin{align*}
&\mathrel{\hphantom{=}}\bO\bigl(m^s \cdot (n^2 \log m^{1 - s}
+ nm^{(1 - s)(2 + \eps)} + nm^{1 - s}\log n)
+ m^{1 + s} \log n \log nm^{1 - s}\bigr)\\
&= \bO(n^2m^s\log m + nm^{(1 - s)(2 + \eps) + s} + nm\log n + m^{1 + s} \log^2 n
+ m^{1 + s} \log n \log m)\,.
\end{align*}
For \(n \gg m\), the time is dominated by \(\bO(n^2m^s\log m)\), so we pick
\(s = 0\) and get \(\bO(n^2\log m + nm^{2 + \eps} + nm\log n)\) time.
For \(n \approx m\) or \(m \gg n\), we balance the powers by picking
\(s = \sfrac{(1 + \eps)}{(2 + \eps)}\) to get
\(\bO(n^2 m^{\sfrac{1}{2} + \eps'}\log m + nm\log n + nm^{\sfrac{3}{2} + \eps'}
+ m^{\sfrac{3}{2} + \eps'}\log n \log m)\) time for this choice of \(s\),
where \(\eps' > 0\) is an arbitrarily small constant.

\subsection{Preprocessing Polygons}\label{sec:ccobjects}
Instead of line segments in the set \(A\), we can extend our approach to
polygons in \(A\).
So consider now the setting where we have a query segment \(\queryseg\) and a
set of polygons \(A\).

For this extension, we mainly have to show that an equivalent to
\cref{lem:vis_glass_cone} holds---the rest then easily follows as for line
segments in \cref{sec:segment_segment}.
We first have to define the visibility glass and the visibility glass cone.
The visibility glass between a diagonal \(D = \seg{uv}\) and a polygon
\(S \in A\) is defined as the set of segments \(\seg{aw}\) where \(a \in S\) and
\(w \in D\).
Without loss of generality, assume that \(D\) is vertical and that it splits
\(P\) into subpolygons \(P_L\) and \(P_R\), left and right of \(D\),
respectively, and assume \(S \subseteq P_L\).

Consider the segment \(\seg{ay}\) of the visibility glass with the highest
slope.
In case of ties, take the shortest segment, so the intersection of \(\seg{ay}\)
and \(S\) consists of only \(a\).
Similarly, define \(\seg{cx}\) as the segment with the lowest slope.
Let \(L_L\) and \(L_R\) denote the supporting lines of \(\seg{ay}\) and
\(\seg{cx}\), respectively.
The visibility glass cone of \(S\) through \(D\) is then the cone defined by
\(L_L\) and \(L_R\) that passes through \(D\).
We are now ready to prove an equivalent to \cref{lem:vis_glass_cone}.

\begin{figure}[tb]
\centering
\includegraphics{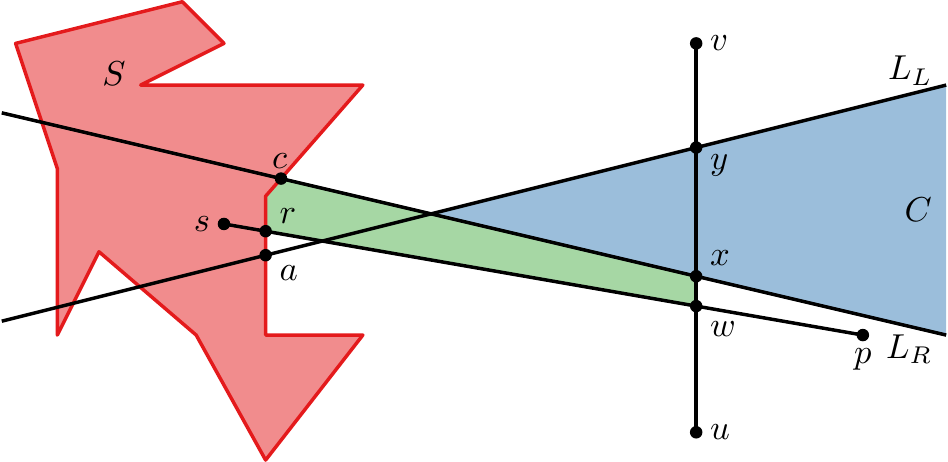}
\caption{The visibility glass cone \(C\) for a region \(S\) and a diagonal
\(D = \seg{uv}\).}
\label{fig:region-cone}
\end{figure}

\begin{lemma}\label{lem:vis_glass_cone_region}
Consider a polygon \(P\), split into subpolygons \(P_L\) and \(P_R\) by a
diagonal \(D = \seg{uv}\) between two vertices \(u\) and \(v\), and let \(S\) be
a simple polygon in \(P_L\).
Let \(C\) be the visibility glass cone of \(S\) into \(P_R\) through \(D\).
If some point \(p \in P_R\) is visible to \(S\), it must be in \(C\).
\end{lemma}
\begin{proof}
For a contradiction, assume that there is a point \(p \in P_R\) that sees \(S\)
but lies outside of \(C\).
Without loss of generality, suppose that \(p\) is below \(C\).
Let \(s \in S\) be a point that is visible to \(p\), and let \(r \in S\) be the
point closest to \(p\) on the line segment \(\seg{sp}\)---see also
\cref{fig:region-cone}.
Let \(w\) be the intersection point of \(\seg{sp}\) and \(D\).
We first note that \(s\) cannot lie above (or on) \(L_R\).
If that were the case, then the line segment \(\seg{ws}\) would have a lower
slope than \(L_R\) and would be in the visibility glass, contradicting the
definition of the visibility glass cone.
So we can assume that \(s\) is below \(L_R\).
However, if \(s\) is below \(L_R\), then so is \(r\).
Now consider the region defined by \(\seg{cx}\), \(\seg{xw}\), \(\seg{wr}\),
and the path from \(c\) to \(r\) along the boundary of \(S\) in clockwise
direction (green region in \cref{fig:region-cone}).
None of these segments or the path can be intersected by the polygon boundary,
so the region is empty.
However, in that case, also the line segment \(\seg{cw}\) must be in the
visibility glass and has a lower slope, again contradicting the definition of
the visibility glass cone.
From this contradiction we can conclude that any \(p \in P_R\) that sees \(S\)
must be inside the visibility glass cone \(C\).
\end{proof}

Using this \lcnamecref{lem:vis_glass_cone_region}, we can apply the same methods
as in \cref{sec:segment_segment} for a set of segments.

\subsection{Moving Points}\label{sec:moving_points}
In the context of moving objects, we may interpret a segment as a moving object
that traverses the segment from start to end with a constant velocity.
This applies both to the objects in a given set \(A\) and the query object.
More formally, consider objects \(p\), \(q\) that have trajectories
\(p(t): \R \to \R^2 \cap P\) and \(q(t): \R \to \R^2 \cap P\) inside a polygon
\(P\).
We say that \(p\) and \(q\) are \emph{mutually visible} in \(P\) if and only if
at some time \(t\), the line segment \(\seg{p(t)q(t)}\) is inside the polygon
\(P\).
In this case, we could be interested in counting how many objects can be seen
by the query object at some point during their movement.
Note that the settings we discuss in \cref{sec:point,sec:point_segment} lend
themselves to this interpretation immediately, since either the query or the
objects of \(A\) do not move.
On the other hand, the setting of a query segment with a set of segments from
\cref{sec:segment_segment} does not translate to moving objects.

Eades et al.~\cite{eades20} present a data structure that supports determining
whether two query objects see each other at some point in time by preprocessing
only the polygon.
There is no obvious extension to their data structure that also preprocesses the
set of objects.
A possible (slow) solution would be to track time as a third dimension and
construct the visibility polygon of each point \(p \in A\) as it moves.
Given a moving object as a query, we would then need to count the visibility
polygons (that include a time dimension) that are pierced by the segment.
It seems difficult to avoid double counting the points in this scenario;
actually solving this problem would be an interesting continuation of the work
presented in this paper.

\subsection{Query Variations}\label{sec:query_variations}
There are many other settings that one could consider as extensions of this
work.
For instance, we could solve the simpler problem of testing visibility: given a
query point \(q\) or line segment \(\queryseg\), check if it sees any object in
the set \(A\).
Surprisingly, it does not seem easy to simplify our approaches to answer this
question more efficiently.
We could also consider the reporting version of the problem rather than
counting; this works immediately for the point query approaches of
\cref{sec:point}, but our use of inclusion--exclusion arguments for segment
queries in \cref{sec:point_segment,sec:segment_segment} prevents us from easily
adapting those to reporting in time proportional to the number of reported
segments.
Finally, when considering segments, one can ask many other questions: how much
of each segment is seen by a query segment and vice versa, for each segment or
in total; these questions and more can also be considered for moving objects as
in \cref{sec:moving_points}.
All of these would be highly exciting directions for future work.

\bibliographystyle{plainurl}
\bibliography{references}

\begin{thebibliography}{10}

\bibitem{agarwal93rayshoot}
Pankaj~K. Agarwal and Ji{\v{r}\'i} Matou{\v{s}}ek.
\newblock Ray shooting and parametric search.
\newblock {\em SIAM Journal on Computing}, 22(4):794--806, 1993.
\newblock \href {https://doi.org/10.1137/0222051} {\path{doi:10.1137/0222051}}.

\bibitem{agarwal93spacepart}
Pankaj~K. Agarwal and Micha Sharir.
\newblock Applications of a new space-partitioning technique.
\newblock {\em Discrete \& Computational Geometry}, 9:11--38, 1993.
\newblock \href {https://doi.org/10.1007/BF02189304}
  {\path{doi:10.1007/BF02189304}}.

\bibitem{agarwal96}
Pankaj~K. Agarwal and Marc~J. van Kreveld.
\newblock Connected component and simple polygon intersection searching.
\newblock {\em Algorithmica}, 15:626--660, 1996.
\newblock \href {https://doi.org/10.1007/BF01940884}
  {\path{doi:10.1007/BF01940884}}.

\bibitem{alipour15}
Sharareh Alipour, Mohammad Ghodsi, Alireza Zarei, and Maryam Pourreza.
\newblock Visibility testing and counting.
\newblock {\em Information Processing Letters}, 115(9):649--654, 2015.
\newblock \href {https://doi.org/10.1016/j.ipl.2015.03.009}
  {\path{doi:10.1016/j.ipl.2015.03.009}}.

\bibitem{aronov02}
Boris Aronov, Leonidas~J. Guibas, Marek Teichmann, and Li~Zhang.
\newblock Visibility queries and maintenance in simple polygons.
\newblock {\em Discrete \& Computational Geometry}, 27:461--483, 2002.
\newblock \href {https://doi.org/10.1007/s00454-001-0089-9}
  {\path{doi:10.1007/s00454-001-0089-9}}.

\bibitem{ben-moshe04}
Boaz Ben-Moshe, Olaf Hall-Holt, Matthew~J. Katz, and Joseph S.~B. Mitchell.
\newblock Computing the visibility graph of points within a polygon.
\newblock In Jack~S. Snoeyink and Jean-Daniel Boissonnat, editors, {\em
  Proceedings of the 20th Annual Symposium on Computational Geometry ({SoCG}
  2004)}, pages 27--35, New York, NY, USA, 2004. ACM.
\newblock \href {https://doi.org/10.1145/997817.997825}
  {\path{doi:10.1145/997817.997825}}.

\bibitem{bose02}
Prosenjit Bose, Anna Lubiw, and James~Ian Munro.
\newblock Efficient visibility queries in simple polygons.
\newblock {\em Computational Geometry: Theory \& Applications}, 23(3):313--335,
  2002.
\newblock \href {https://doi.org/10.1016/S0925-7721(01)00070-0}
  {\path{doi:10.1016/S0925-7721(01)00070-0}}.

\bibitem{bygi15}
Mojtaba~Nouri Bygi, Shervin Daneshpajouh, Sharareh Alipour, and Mohammad
  Ghodsi.
\newblock Weak visibility counting in simple polygons.
\newblock {\em Journal of Computational and Applied Mathematics}, 288:215--222,
  2015.
\newblock \href {https://doi.org/10.1016/j.cam.2015.04.018}
  {\path{doi:10.1016/j.cam.2015.04.018}}.

\bibitem{chazelle82}
Bernard Chazelle.
\newblock A theorem on polygon cutting with applications.
\newblock In {\em Proceedings of the 23rd Annual {IEEE} Symposium on
  Foundations of Computer Science ({FOCS} 1982)}, pages 339--349, Piscataway,
  NJ, USA, 1982. IEEE.
\newblock \href {https://doi.org/10.1109/SFCS.1982.58}
  {\path{doi:10.1109/SFCS.1982.58}}.

\bibitem{chazelle91}
Bernard Chazelle.
\newblock Triangulating a simple polygon in linear time.
\newblock {\em Discrete \& Computational Geometry}, 6:485--524, 1991.
\newblock \href {https://doi.org/10.1007/BF02574703}
  {\path{doi:10.1007/BF02574703}}.

\bibitem{chazelle93}
Bernard Chazelle.
\newblock Cutting hyperplanes for divide-and-conquer.
\newblock {\em Discrete \& Computational Geometry}, 9:145--158, 1993.
\newblock \href {https://doi.org/10.1007/BF02189314}
  {\path{doi:10.1007/BF02189314}}.

\bibitem{chazelle92intls}
Bernard Chazelle and Herbert Edelsbrunner.
\newblock An optimal algorithm for intersecting line segments in the plane.
\newblock {\em Journal of the {ACM}}, 39(1):1--54, 1992.
\newblock \href {https://doi.org/10.1145/147508.147511}
  {\path{doi:10.1145/147508.147511}}.

\bibitem{chazelle89}
Bernard Chazelle and Leonidas~J. Guibas.
\newblock Visibility and intersection problems in plane geometry.
\newblock {\em Discrete \& Computational Geometry}, 4:551--581, 1989.
\newblock \href {https://doi.org/10.1007/BF02187747}
  {\path{doi:10.1007/BF02187747}}.

\bibitem{chazelle92fast}
Bernard Chazelle, Micha Sharir, and Emo Welzl.
\newblock Quasi-optimal upper bounds for simplex range searching and new zone
  theorems.
\newblock {\em Algorithmica}, 8:407--429, 1992.
\newblock \href {https://doi.org/10.1007/BF01758854}
  {\path{doi:10.1007/BF01758854}}.

\bibitem{chen15weak}
Danny~Ziyi Chen and Haitao Wang.
\newblock Weak visibility queries of line segments in simple polygons.
\newblock {\em Computational Geometry: Theory \& Applications}, 48(6):443--452,
  2015.
\newblock \href {https://doi.org/10.1016/j.comgeo.2015.02.001}
  {\path{doi:10.1016/j.comgeo.2015.02.001}}.

\bibitem{clarkson87}
Kenneth~L. Clarkson.
\newblock New applications of random sampling in computational geometry.
\newblock {\em Discrete \& Computational Geometry}, 2:195--222, 1987.
\newblock \href {https://doi.org/10.1007/BF02187879}
  {\path{doi:10.1007/BF02187879}}.

\bibitem{berg94}
Mark de~Berg, Dan Halperin, Mark~H. Overmars, Jack~S. Snoeyink, and Marc~J. van
  Kreveld.
\newblock Efficient ray shooting and hidden surface removal.
\newblock {\em Algorithmica}, 12:30--53, 1994.
\newblock \href {https://doi.org/10.1007/BF01377182}
  {\path{doi:10.1007/BF01377182}}.

\bibitem{eades20}
Patrick Eades, Ivor van~der Hoog, Maarten L{\"o}ffler, and Frank Staals.
\newblock Trajectory visibility.
\newblock In Susanne Albers, editor, {\em Proceedings of the 17th Scandinavian
  Symposium and Workshops on Algorithm Theory ({SWAT} 2020)}, number 162 in
  Leibniz International Proceedings in Informatics ({LIPIcs}), pages
  23:1--23:22, Dagstuhl, Germany, 2020. Schloss Dagstuhl~--~Leibniz-Zentrum
  f{\"u}r Informatik.
\newblock \href {https://doi.org/10.4230/LIPIcs.SWAT.2020.23}
  {\path{doi:10.4230/LIPIcs.SWAT.2020.23}}.

\bibitem{elgindy81}
Hossam El~Gindy and David Avis.
\newblock A linear algorithm for computing the visibility polygon from a point.
\newblock {\em Journal of Algorithms}, 2(2):186--197, 1981.
\newblock \href {https://doi.org/10.1016/0196-6774(81)90019-5}
  {\path{doi:10.1016/0196-6774(81)90019-5}}.

\bibitem{ghosh07}
Subir~Kumar Ghosh.
\newblock {\em Visibility Algorithms in the Plane}.
\newblock Cambridge University Press, Cambridge, UK, 2007.
\newblock \href {https://doi.org/10.1017/CBO9780511543340}
  {\path{doi:10.1017/CBO9780511543340}}.

\bibitem{gudmundsson10}
Joachim Gudmundsson and Pat Morin.
\newblock Planar visibility: Testing and counting.
\newblock In David~G. Kirkpatrick and Joseph S.~B. Mitchell, editors, {\em
  Proceedings of the 26th Annual Symposium on Computational Geometry ({SoCG}
  2010)}, pages 77--86, New York, NY, USA, 2010. ACM.
\newblock \href {https://doi.org/10.1145/1810959.1810973}
  {\path{doi:10.1145/1810959.1810973}}.

\bibitem{guibas89}
Leonidas~J. Guibas and John Hershberger.
\newblock Optimal shortest path queries in a simple polygon.
\newblock {\em Journal of Computer and System Sciences}, 39(2):126--152, 1989.
\newblock \href {https://doi.org/10.1016/0022-0000(89)90041-X}
  {\path{doi:10.1016/0022-0000(89)90041-X}}.

\bibitem{guibas87}
Leonidas~J. Guibas, John Hershberger, Daniel Leven, Micha Sharir, and Robert~E.
  Tarjan.
\newblock Linear-time algorithms for visibility and shortest path problems
  inside triangulated simple polygons.
\newblock {\em Algorithmica}, 2:209--233, 1987.
\newblock \href {https://doi.org/10.1007/BF01840360}
  {\path{doi:10.1007/BF01840360}}.

\bibitem{gupta95}
Prosenjit Gupta, Ravi Janardan, and Michiel H.~M. Smid.
\newblock Further results on generalized intersection searching problems:
  Counting, reporting, and dynamization.
\newblock {\em Journal of Algorithms}, 19(2):282--317, 1995.
\newblock \href {https://doi.org/10.1006/jagm.1995.1038}
  {\path{doi:10.1006/jagm.1995.1038}}.

\bibitem{hershberger91}
John Hershberger.
\newblock A new data structure for shortest path queries in a simple polygon.
\newblock {\em Information Processing Letters}, 38(5):231--235, 1991.
\newblock \href {https://doi.org/10.1016/0020-0190(91)90064-O}
  {\path{doi:10.1016/0020-0190(91)90064-O}}.

\bibitem{hershberger95}
John Hershberger and Subhash Suri.
\newblock A pedestrian approach to ray shooting: Shoot a ray, take a walk.
\newblock {\em Journal of Algorithms}, 18(3):403--431, 1995.
\newblock \href {https://doi.org/10.1006/jagm.1995.1017}
  {\path{doi:10.1006/jagm.1995.1017}}.

\bibitem{joe87}
Barry Joe and Richard~B. Simpson.
\newblock Corrections to {L}ee's visibility polygon algorithm.
\newblock {\em BIT Numerical Mathematics}, 27:458--473, 1987.
\newblock \href {https://doi.org/10.1007/BF01937271}
  {\path{doi:10.1007/BF01937271}}.

\bibitem{kirkpatrick83}
David~G. Kirkpatrick.
\newblock Optimal search in planar subdivisions.
\newblock {\em SIAM Journal on Computing}, 12(1):28--35, 1983.
\newblock \href {https://doi.org/10.1137/0212002} {\path{doi:10.1137/0212002}}.

\bibitem{lee83}
Der-Tsai Lee.
\newblock Visibility of a simple polygon.
\newblock {\em Computer Vision, Graphics, and Image Processing},
  22(2):207--221, 1983.
\newblock \href {https://doi.org/10.1016/0734-189X(83)90065-8}
  {\path{doi:10.1016/0734-189X(83)90065-8}}.

\bibitem{orourke87}
Joseph O'Rourke.
\newblock {\em Art Gallery Theorems and Algorithms}, volume~3 of {\em The
  International Series of Monographs on Computer Science}.
\newblock Oxford University Press, Oxford, UK, 1987.
\newblock URL:
  \url{http://www.science.smith.edu/~jorourke/books/ArtGalleryTheorems/art.html}.

\bibitem{overmars88}
Mark~H. Overmars and Emo Welzl.
\newblock New methods for computing visibility graphs.
\newblock In Herbert Edelsbrunner, editor, {\em Proceedings of the 4th Annual
  Symposium on Computational Geometry ({SoCG} 1988)}, pages 164--171, New York,
  NY, USA, 1988. ACM.
\newblock \href {https://doi.org/10.1145/73393.73410}
  {\path{doi:10.1145/73393.73410}}.

\bibitem{suri86}
Subhash Suri and Joseph O'Rourke.
\newblock Worst-case optimal algorithms for constructing visibility polygons
  with holes.
\newblock In Alok Aggarwal, editor, {\em Proceedings of the 2nd Annual
  Symposium on Computational Geometry ({SoCG} 1986)}, pages 14--23, New York,
  NY, USA, 1986. ACM.
\newblock \href {https://doi.org/10.1145/10515.10517}
  {\path{doi:10.1145/10515.10517}}.

\end{thebibliography}
\end{document}